\documentclass[runningheads]{llncs}

\usepackage[T1]{fontenc}
\usepackage[margin=1in]{geometry}
\usepackage{algpseudocode}
\usepackage{xcolor}
\usepackage{graphicx}
\usepackage{xspace}
\usepackage[linesnumbered,ruled,vlined,noend]{algorithm2e}
\NoCaptionOfAlgo
\usepackage{float}
\usepackage{enumitem}
\usepackage{url}
\usepackage{hyperref}
\usepackage{amsmath}
\usepackage{amssymb}
\usepackage{multirow}
\usepackage{booktabs}
\usepackage{comment}

\newcommand{\Con}{\texttt{Con}\xspace}
\newcommand{\defn}[1]{\textbf{\emph{#1}}}
\newcommand{\CollCost}{\mathcal{C}\xspace}
\newcommand{\currWin}{w_{\mbox{\tiny cur}}\xspace}
\newcommand{\poly}[1]{{\texttt{poly}({#1})}}

\newcommand{\tunableparam}{\epsilon\xspace}
\newcommand{\mainAlgorithm}{\textsc{Lowball}\xspace}
\newcommand{\mainAlgorithmAcrynom}{\textsc{Lowball}\xspace}
\newcommand{\whp}{w.h.p.\xspace}

\newif\ifcomments
\commentsfalse

\newif\ifshowproofs
\showproofstrue

\AtBeginDocument{}

\title{Dynamic Wakeup under Costly Collisions\vspace{-8pt}}

\author{Umesh Biswas \and
Maxwell Young}
\authorrunning{U. Biswas and M. Young}
\institute{
Department of Computer Science and Engineering,
Mississippi State University, MS 39762, USA\\
\email{ucb5@msstate.edu, myoung@cse.msstate.edu}
}

\begin{document}
\maketitle

\begin{abstract}
The wakeup problem captures a fundamental symmetry-breaking challenge
among devices sharing a communication channel. We study the dynamic
setting, where packets become active at arbitrary times on a time-slotted
multiple access channel.  In each slot, a transmission succeeds if and only if exactly one packet
transmits; two or more simultaneous transmissions cause a collision. The goal is to obtain a successful
transmission quickly.

Prior work on wakeup has largely focused on the number of slots until
the first success; that is, the latency. However, a collision may incur
substantial additional delay, represented by a per-collision cost
$\CollCost$. We therefore seek to control both latency and the collision
cost of an execution, defined as $\CollCost$ times its number of collisions.

We design and analyze a randomized algorithm for dynamic wakeup,
\mainAlgorithm, without collision detection or
knowledge of the number of packets, $n$. Fix a constant $0<\tunableparam\le1/2$. There is a
constant $K>0$ such that, when
$\CollCost\ge K\lg^{1/\tunableparam}n$, \mainAlgorithmAcrynom has
expected latency
$O(\CollCost^{1/2+\tunableparam}\ln\CollCost)$ and expected collision
cost $O(\sqrt{\CollCost})$. Below this threshold, both expectations
are $O(n\log^{\Theta(1/\tunableparam)}n)$. These guarantees hold
against an adaptive adversary, and the
algorithm succeeds with probability $1$. For algorithms in which each
packet's transmission probability depends only on $\CollCost$ and the
packet's local age, with packets activated together using the same
probability schedule, we prove that the maximum of expected latency
and expected collision cost is $\Omega(\sqrt{\CollCost})$.
\end{abstract}

\medskip

\section{Introduction}

In the wakeup problem, devices are classified as either \defn{active} ({\it awake}) or \defn{inactive} ({\it asleep}). Active devices have packets ready for transmission over a shared channel, whereas inactive devices only listen. The first successful transmission serves as a “wakeup call”, awakening all inactive devices and resolving the problem \cite{banicescu2024survey}.

Traditionally, the performance of a wakeup algorithm is measured by the number of slots until this first successful transmission, referred to as the \defn{latency}. Achieving small latency is nontrivial because the channel can support at most one successful sender in a slot: if two or more devices transmit simultaneously, a \defn{collision} occurs and all of those transmissions fail. Moreover, the devices operate in a distributed manner, with no {\it a priori} central scheduler or coordinator to assign transmission times.

In this work, we address both latency and the additional delay caused by
collisions in solving the wakeup problem. The per-collision cost
$\CollCost$ represents the delay from one collision. Summing these
charges separately from the slotted execution gives the execution's
collision cost. While existing algorithms give
strong latency guarantees, they may incur a large number of collisions (see
Section~\ref{s:baseline}), and addressing these two metrics in tandem
requires new ideas. In particular, there is a natural tension between them.
An aggressive transmission strategy can improve latency by increasing the
chance that some packet sends, but it can also produce many costly collisions.
On the other hand, a conservative strategy can reduce the delay caused by
collisions, but may postpone the first successful transmission. Thus,
minimizing the number of slots to first success does not necessarily
minimize the accumulated delay caused by collisions.

This tension is even more pronounced in the dynamic wakeup setting. Packets
may become active at arbitrary times and therefore execute according to
different local clocks. Moreover, the packets do not know the number of packets {\boldmath{$n$}}, and without
collision detection they cannot determine whether the current contention is
too low, too high, or already at an appropriate level. Consequently, an
algorithm must regulate contention without direct feedback while remaining
robust to adversarially timed packet activations. This raises the central
question of our work: {\it Can we solve dynamic wakeup without collision
detection while keeping both latency and collision cost small?}
\medskip

\noindent{\bf Why Collision Cost Matters.} In the standard wakeup model, a collision consumes a single slot.
In practice, however, a collision can occupy the channel or trigger
a recovery period that causes substantial additional delay. For
example, in WiFi (IEEE 802.11), the time consumed by a transmission
and recovery from a collision can grow with packet size
\cite{anderton:windowed,anderton:is,biswas2026dynamic,biswas2026softening}.
An algorithm that obtains a success after few contention opportunities
may therefore still incur substantial elapsed delay: even a few costly
collisions can outweigh the savings from using fewer slots.

This concern extends beyond short-range WiFi. On long-distance
wireless links, acknowledgment timeouts must accommodate propagation
delay, increasing the time spent waiting after a failed transmission
\cite{patra2007wildnet}. In intermittently connected mobile networks
\cite{zhang2006routing}, a collision may force the sender to wait for
another communication opportunity. In such settings, treating a
collision as a single slot can substantially understate its effect
on the time to communicate successfully.

We represent this additional delay by the per-collision cost
$\CollCost$. We report latency and the collision cost of the execution
separately, since reducing one can increase the other. Together, these
metrics capture the slots through the first success and the additional
delay represented by the accumulated collision charges.

The per-collision cost need not remain constant as the system grows.
Packet sizes may increase with the amount of control information or
application data they carry, increasing the delay associated with a
collision \cite{anderton:windowed}. Likewise, if additional devices are
accommodated by expanding the physical area served by the shared
channel, longer propagation distances can require larger recovery
timeouts \cite{patra2007wildnet}. We therefore allow $\CollCost$ to vary
with $n$, without assuming a specific dependence.

%For Bluetooth \cite{ghamari2018detailed}, collisions waste both time, in addition to e and resources, resulting in degraded network performance.

%%%%%%%%%%%%%%%%%%%%%%%%%%%%%%%%%%%
%%%%%%%%%%%%%%%%%%%%%%%%%%%%%%%%%%%
%%%%%%%%%%%%%%%%%%%%%%%%%%%%%%%%%%%

\subsection{Model and Notation}\label{sec:model}

We consider a system with {\boldmath{$n$}} devices, for a large and unknown value of $n$; each device has a single packet to send on the shared channel. We consider executions in which at least one packet is activated. Packets have only local clocks, which start at activation, and they do not have identifiers. Going forward, for ease of presentation, we refer only to the packets---rather than devices---committing actions such as sending and listening. \medskip 

\noindent{\bf Communication.} Time is divided into disjoint \defn{slots}, each of which can accommodate a packet transmission. Communication occurs on a \defn{multiple access channel}, which is defined as follows.  In any slot, a packet can either attempt to transmit or listen to the channel. If no packet is transmitted in a given slot, it is \defn{empty}. If only one packet is transmitted, the packet \defn{succeeds}; we often refer to this as a \defn{success}. However, if multiple packets transmit simultaneously, they all fail---this is a \defn{collision}---and they may try again later. There is no scheduler or central authority {\it a priori}.\smallskip

\noindent{\bf No Collision Detection.} There is no mechanism to distinguish between an empty slot and a collision—that is, {\it no collision detection} (no-CD). In many practical systems, collision detection (CD) is unavailable or unreliable. For example, in WiFi networks, received signal strength may offer a weak form of CD, but it is error-prone. Consequently, many prior results address the challenging no-CD setting (see the survey \cite{banicescu2024survey} and related discussion in Section \ref{sec:related-work}).

We emphasize that CD and per-collision cost are separate concerns: collisions can cause additional delay regardless of whether they are detectable. For example, in WiFi, even if devices cannot sense an ongoing collision, the channel remains unavailable for its duration.
\medskip

\noindent{\bf Adversarial Activations.} Packet activation times are controlled by an adaptive adversary. Before each slot $t$, the adversary may choose which packets become active in that slot as an arbitrary function of the entire execution history through the end of slot $t-1$, including all previous packet actions and channel outcomes. The adversary must commit to these activations before any random choices or transmission decisions for slot $t$ are revealed. Thus, its activation decision for slot $t$ may depend on the history through slot $t-1$, but not on random bits or packet actions in slot $t$. After slot $t$ is completed, the information revealed by that slot becomes part of the execution history and may be used by the adversary when determining future activations. This adaptive-adversary convention is consistent with prior work on dynamic wakeup, where packets may be activated at adversarially chosen times \cite{gkasieniec2000wakeup,jurdzinski2002probabilistic,jurdzinski2015cost}, and with closely-related dynamic contention-resolution models that allow adaptive arrivals without advance knowledge of fresh random choices \cite{DeMarco:2017:ASC:3087801.3087831,bender:fully,BenderFGY19,BenderKPY18}.\medskip

\noindent{\bf Metrics.} We measure \defn{latency} by the number of slots
from the first packet activation through the first success. An execution
with no success has infinite latency. Each collision incurs a
\underline{known} \defn{per-collision cost} {\boldmath{$\CollCost$}},
which is fixed throughout the execution and at least a sufficiently large
constant.\footnote{When $\CollCost=O(1)$, an execution's collision cost
is at most a constant times its latency.} To capture worst-case
performance, the adversary chooses $\CollCost$ before execution with
knowledge of $n$,
subject to a polynomial upper bound $\CollCost=O(\poly{n})$ whose degree
is fixed but unknown to the packets. The packets know $\CollCost$, but
are not given $n$ or an estimate of $n$; their behavior may therefore
depend on $\CollCost$.

The \defn{collision cost} of an execution is $\CollCost$ times the
number of collision slots before its first success, counting all
collision slots if it never succeeds. Each such slot is charged once,
not once per transmitting packet. Expected collision cost is the
expectation of this sum; for specified slots, we sum only their charges.
These charges represent additional delay but are accounted for
separately: they neither insert extra slots nor provide feedback, and
packet actions and activations retain their usual slot indices.

We report worst-case expected latency and expected collision cost
separately, over adversaries satisfying the model, with the objective
of minimizing their maximum.
\medskip

\noindent{\bf Notation.} An event holds with high probability (\defn{\whp}) in $x$ if its failure probability is $O(x^{-a})$
for any fixed constant $a>0$, with the constants in the algorithm
chosen accordingly. We use $\lg$ for base-$2$ logarithms and $\ln$
for natural logarithms; in asymptotic bounds, $\log$ may use any
fixed base greater than $1$. We write $\log^c y=(\log y)^c$.

Our algorithm is parameterized by a fixed constant
$0<\tunableparam\le1/2$, chosen before execution. Constants
hidden in asymptotic notation may depend on $\tunableparam$.
For fixed constants $k\ge0$ and $k'>0$ independent of
$\tunableparam$ and $n$, we abbreviate factors of the form
$\log^{k+k'/\tunableparam}n$ as
$\log^{\Theta(1/\tunableparam)}n$.

%%%%%%%%%%%%%%%%%%%%%%%%%%%%%%%%%%%
%%%%%%%%%%%%%%%%%%%%%%%%%%%%%%%%%%%
%%%%%%%%%%%%%%%%%%%%%%%%%%%%%%%%%%%

\subsection{Our Results}

%\umesh{A comparison between the static and dynamic results for Case 1 is required.}

We design and analyze our algorithm, \mainAlgorithm,
in the dynamic wakeup setting. Fix a constant $0<\tunableparam\le1/2$ and choose the chunk-size constant $d$ sufficiently large. Our bounds separate according to the per-collision cost $\CollCost$. When $\CollCost$ exceeds a
polylogarithmic threshold in $n$, \mainAlgorithmAcrynom achieves expected collision cost
$O(\sqrt{\CollCost})$ and expected latency $O(\CollCost^{1/2+\tunableparam}\ln\CollCost)$. Below this threshold,
both expectations are $O(n\log^{\Theta(1/\tunableparam)}n)$.\medskip

\noindent{\bf Upper Bound.} Our upper-bound result is the following.

\begin{theorem}\label{thm:dynamic}
Fix a constant $0<\tunableparam\le1/2$ and choose $d$
sufficiently large. There is a constant $K>0$, depending on
$\tunableparam$, such that against any adaptive adversary, \mainAlgorithmAcrynom solves dynamic wakeup with the following bounds.
\begin{enumerate}[leftmargin=15pt]
    \item \textbf{Large-$\CollCost$ regime.}
    If $\CollCost\ge K \lg^{1/\tunableparam}n$, then the
    expected latency is
    $O(\CollCost^{1/2+\tunableparam}\ln\CollCost)$ and the
    expected collision cost is $O(\sqrt{\CollCost})$.

    \item \textbf{Small-$\CollCost$ regime.}
    If $\CollCost<K\lg^{1/\tunableparam} n$, then the
    expected latency and expected collision cost are both
    $O(n\log^{\Theta(1/\tunableparam)}n)$.
\end{enumerate}
\end{theorem}

We emphasize that \mainAlgorithmAcrynom does not require $n$ as input
and does not need to determine which case of the theorem applies.

%%%%%%%%%%%%%%
Classical dynamic-wakeup algorithms provide strong latency guarantees in
closely related no-CD models. Jurdzi\'nski and Stachowiak
\cite{jurdzinski2005probabilistic} give an $O(n/\log n)$ latency bound
for fixed error probability when $n$ is unknown. More recently, the
\textsc{DecreaseSlowly} wakeup protocol analyzed by De Marco et
al.~\cite{de2022contention} achieves $O(k)$ wakeup time \whp, where
$k$ is the number of packets activated during the relevant activity
interval, even against an adaptive adversary.

These results focus on obtaining small latency rather than controlling
the cost of collisions. For example, \textsc{DecreaseSlowly} begins
with transmission probability $1/2$. Thus, if two packets begin the
protocol together, there is already a constant probability of collision
in the first transmission round, resulting in $\Omega(\CollCost)$
expected collision cost from that round alone.

Our objective is different: we aim to keep both latency and collision
cost small. In the large-$\CollCost$ regime, \mainAlgorithmAcrynom
achieves expected collision cost $O(\sqrt{\CollCost})$ together with
expected latency $O(\CollCost^{1/2+\tunableparam}\ln\CollCost)$.
\medskip

\noindent{\bf Lower Bound.} We also establish a lower bound for
\defn{age-based batch-fair} algorithms: each packet transmits
independently with a probability that depends only on $\CollCost$
and its local age. In particular, packets activated together use the
same sending probability in each slot. \mainAlgorithmAcrynom belongs
to this class. For an algorithm $A$, let
$\mathcal{L}_A(n,\CollCost)$ and
$\mathcal{K}_A(n,\CollCost)$ denote its worst-case expected latency
and expected collision cost, respectively, over permitted adversaries
activating at most $n$ packets. Define:
$$
\mathcal{M}_A(n,\CollCost)
:=
\max\left\{
\mathcal{L}_A(n,\CollCost),
\mathcal{K}_A(n,\CollCost)
\right\}.
$$

\begin{theorem}\label{thm:dynamic-lower-bound}
For every age-based batch-fair algorithm $A$, every $n\ge2$, and every
$\CollCost$,  $\mathcal{M}_A(n,\CollCost)
\ge
{\sqrt{\CollCost}}/{2}$.
\end{theorem}

Consequently, when
$\CollCost\ge K(\lg^{1/\tunableparam} n)$, the maximum of
\mainAlgorithmAcrynom's two expectations is within a factor
$O(\CollCost^\tunableparam\ln\CollCost)$ of the best possible for this
algorithm class. The underlying tradeoff also shows that an
$O(\sqrt{\CollCost})$ worst-case expected collision-cost guarantee
requires $\Omega(\sqrt{\CollCost})$ worst-case expected latency
(Section~\ref{sec:lower-bound}).

The classical $\Omega(n/\log n)$ dynamic-wakeup lower bound of
Jurdzi\'nski and Stachowiak~\cite{jurdzinski2005probabilistic} requires
care when applied here. When $\CollCost$ is held constant while $n$
grows, their construction can be adapted to give an
$\Omega(n/\log n)$ expected-latency lower bound for our algorithm
class. However, the hidden constant and the required threshold on $n$
may depend on $\CollCost$. Because packets know $\CollCost$, their
sending probabilities may change when $\CollCost$ changes, even
though $n$ remains unknown. Thus, this argument does not give an
$\Omega(n/\log n)$ lower bound with a hidden constant independent of
$\CollCost$ when $\CollCost$ varies with $n$, and therefore does not
contradict our large-$\CollCost$ upper bound. We make this
distinction explicit in Section~\ref{sec:lower-bound}.

%%%%%%%%%%%%%%%%%%%%%%%%%%%%%%%%%%%
%%%%%%%%%%%%%%%%%%%%%%%%%%%%%%%%%%%
%%%%%%%%%%%%%%%%%%%%%%%%%%%%%%%%%%%

\section{Related Work}\label{sec:related-work}

% \maxwell{please go ahead and implement this.}
% \umesh{It seems to me that we should cite the following papers:}
% \begin{itemize}
%     \item {De Marco (2005): ``Faster Deterministic Wakeup in Multiple Access
%     Channels.''}
% \end{itemize}

% \umesh{If we consider the similar model, then we should cite the following contention resolution papers as well:}
% \begin{itemize}

%     \item {Yu (2012): ``Dynamic Contention Resolution in Multiple-Access Channels.''}
%     \item {De Marco (2022): ``Contention Resolution Without Collision Detection: Constant Throughput And Logarithmic Energy.''}
%     \item {De Marco (2023): ``Deterministic Non-Adaptive Contention Resolution on a Shared Channel.''}
%     \item {Bender (2020): ``Contention Resolution without Collision Detection.''}
% \end{itemize}

% \noindent{\bf Static and Dynamic Wakeup.}
The wakeup problem has been studied extensively in both static and dynamic settings \cite{chlebus2001randomized,banicescu2024survey}. In the static setting, where all packets are initially active, Willard~\cite{willard:loglog} gives an upper bound of $\lg\lg n+O(1)$ for fair algorithms with collision detection, together with a matching lower bound. Without collision detection, Kushilevitz and Mansour~\cite{kushilevitz1998omega,kushilevitz1993omega} establish an $\Omega(\log n)$ expected-latency lower bound, while Newport~\cite{newport2014radio,newport:radio-journal} gives $\Omega(\log n)$ expected latency and $\Omega(\log^2 n)$ latency \whp for general randomized algorithms. Bar-Yehuda et al.~\cite{bar1992time} provide matching upper bounds of $O(\log n)$ expected latency and $O(\log^2 n)$ latency \whp. Biswas and Young~\cite{biswas2026gentle} study static wakeup with per-collision cost $\CollCost$ and give an algorithm with expected latency $O(\CollCost^{1/2+2\epsilon})$ and expected collision cost $O(\CollCost^{1/2+\epsilon})$ when $\CollCost$ is sufficiently large with respect to $n$; otherwise, both costs are polylogarithmic in $n$.

The dynamic setting is substantially more challenging because packets may become active at arbitrary times and execute according to different local clocks. Gasieniec et al.~\cite{gkasieniec2000wakeup} study wakeup under different synchronization assumptions, including settings in which processors begin their local clocks at different activation times. De Marco et al.~\cite{de2005faster,de2007faster} study deterministic wakeup in the
locally synchronous model, where no global clock is available, processors may
be activated at different times, and the number of processors is unknown.

Jurdziński and Stachowiak~\cite{jurdzinski2015cost} study dynamic
wakeup on the classical multiple-access channel when stations may
become active at arbitrary times and do not have access to a global
clock or collision detection. They establish an
$\Omega(\log^2 n)$ expected-latency lower bound for Las Vegas algorithms, even when stations know the exact value of $n$ and have distinct identifiers, contrasting with the $\Omega(\log n)$ lower bound when all stations begin simultaneously. They also describe a matching $O(\log^2 n)$ Las Vegas upper bound
under the assumption that stations know a polynomial upper bound on $n$. In earlier work, Jurdziński and Stachowiak~\cite{jurdzinski2002probabilistic} study the effects of global clocks, knowledge of $n$, and packet identifiers. Most relevant to our setting, when $n$ is unknown and there is no global clock,
collision detection, or packet identifiers, they give an algorithm with fixed constant error probability $\varepsilon$ and latency $O(n/\log n)$; their nearly linear lower bound also applies in this weaker-information setting.

Wakeup has also been studied under substantial modifications
to the standard communication model. Under the SINR model and assuming knowledge of a
polynomial upper bound on $n$, Jurdziński and
Stachowiak~\cite{jurdzinski2015cost} obtain Monte Carlo and Las Vegas
bounds of $O(\log^2 n/\log\log n)$, while a deterministic algorithm
requires $O(\log^2 n)$ latency~\cite{jurdzinski2013distributed}.
Chlebus et al.~\cite{Chlebus:2016:SWM:2882263.2882514} study dynamic
wakeup over multiple channels and obtain latency
$O(n^{1/b}\ln(1/\varepsilon))$ using $b>1$ channels. \medskip

\noindent{\bf Preliminary work.}
Closely related preliminary results appeared at MUSICAL 2026
\cite{biswas2026gentle} and in the DSN 2026 Doctoral Forum
\cite{biswas2026dynamic}. The MUSICAL paper studies the static setting,
while the two-page DSN Doctoral Forum paper reports an early approach
to dynamic wakeup. The present paper substantially strengthens these results by providing
a full dynamic algorithm and analysis that establish the guarantees of
Theorem~\ref{thm:dynamic}.

Finally, we highlight work on the closely related problem of contention
resolution, where wakeup requires only one successful transmission,
whereas contention resolution requires all active packets to eventually
succeed (see \cite{chlebus2001randomized,banicescu2024survey}). Several
dynamic contention-resolution works consider models quite similar to ours,
including asynchronous arrivals, no collision detection, and limited
knowledge of the number of devices. Yu et al.~\cite{yu2012dynamic} study
adversarial asynchronous arrivals without a global clock or collision
detection, while Bender et al.~\cite{bender2020contention} and De Marco et
al.~\cite{de2022contention} obtain strong throughput guarantees in dynamic
settings without collision detection. Of particular relevance to wakeup,
the \textsc{DecreaseSlowly} protocol analyzed by De Marco et
al.~\cite{de2022contention} achieves $O(k)$ wakeup time \whp under adaptive
activations, but does not seek to control collision cost. De Marco et
al.~\cite{de2023deterministic} further study deterministic non-adaptive
contention resolution under asynchronous activations.

In the setting with nontrivial per-collision cost, Biswas
et al.~\cite{biswas2024softening,biswas2026softening} solve the static
contention-resolution problem with latency and expected collision cost
$\tilde{O}(n\sqrt{\CollCost})$; however, they do not address the dynamic
setting. Moreover, their algorithm assumes collision detection and relies
on this feedback to adjust sending probabilities, which is unavailable in
our model. Anderton et al.~\cite{anderton:windowed} further show that
collisions can significantly degrade the performance of common backoff
algorithms for contention resolution.

\begin{figure}[t]
    \centering
    \includegraphics[
        width=1\linewidth,
        height=0.4\textheight,
        keepaspectratio
    ]{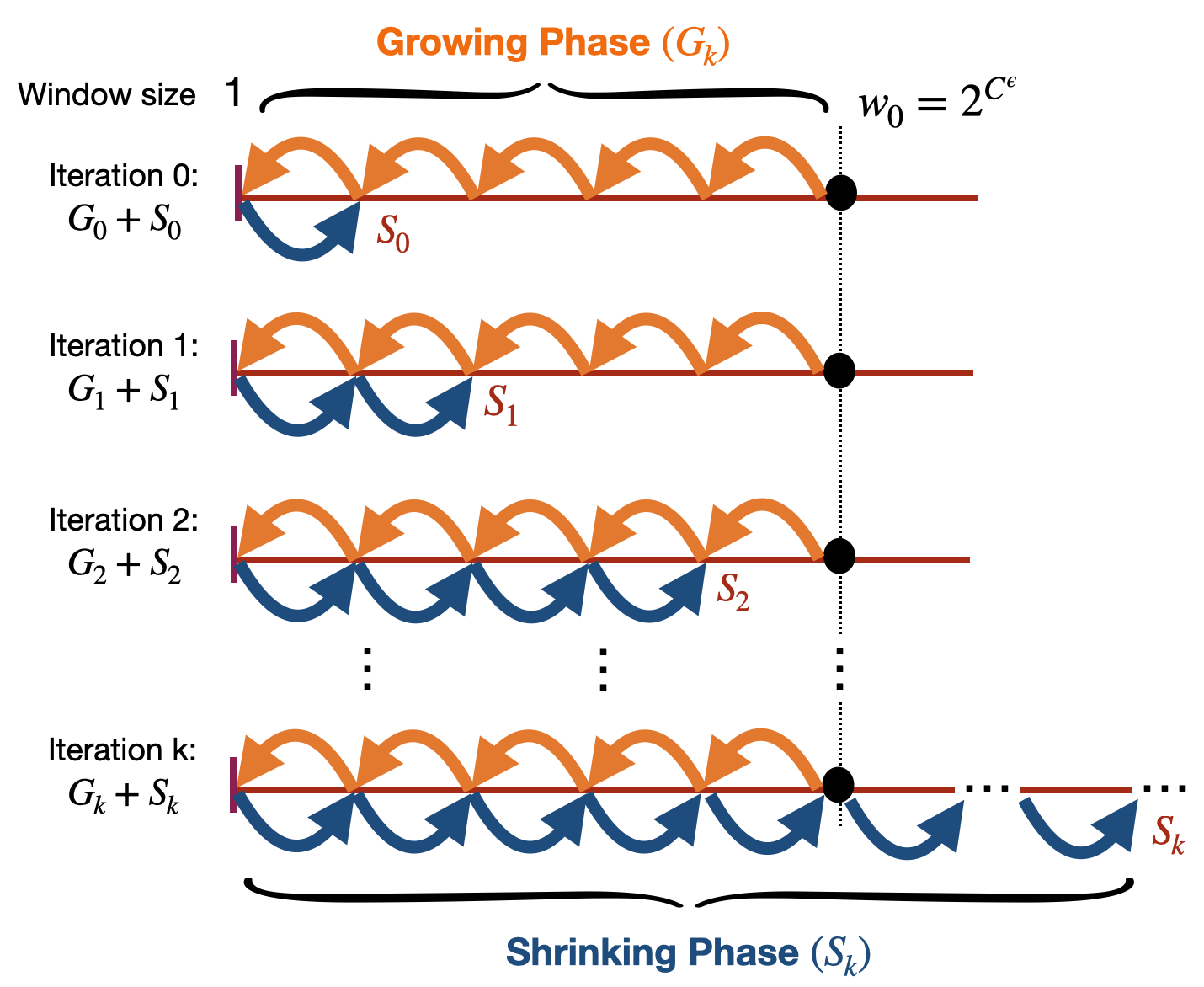}\vspace{-8pt}
    \caption{Illustration of some iterations of \mainAlgorithmAcrynom. Each iteration consists of a fixed-length growing phase $G_k$ (in orange) followed by a shrinking phase $S_k$ (in blue) containing $2^k$ chunks. The top arrows illustrate the common growing-phase schedule, while the bottom arrows illustrate the shrinking phases across successive iterations. Thus, the growing-phase length remains fixed, while the shrinking-phase length doubles across iterations.} %Red, green, and blue indicate the high-, good-, and low-contention regimes, respectively.}
    \label{fig:algorithm}
\end{figure}

\section{Algorithm Specification}\label{sec:our_algo_static}

Set $w_0=2^{\CollCost^\tunableparam}$. We call the algorithm
\mainAlgorithm because it deliberately ``lowballs'' the initial sending
probability, starting each packet at $1/w_0=2^{-\CollCost^\tunableparam}$
to limit the risk of costly early collisions.

The execution proceeds through a sequence of iterations. Each
\defn{iteration} $k$ consists of two phases (as shown in
Figure~\ref{fig:algorithm}): a \defn{growing phase} {\boldmath{$G_k$}}, followed
by a \defn{shrinking phase} {\boldmath{$S_k$}}. These names refer to the batch's
contribution to contention. During $G_k$, each packet's sending
probability increases geometrically, so the batch's contribution to
contention grows. During $S_k$, each packet's sending probability
decreases geometrically, so the batch's contribution to contention
shrinks. The aggregate contention in the system need not be monotone,
since different batches may be in different phases.

More specifically, during $G_k$, the sending probability increases
from $1/w_0$ to a value in $(1/2,1]$. During $S_k$, the sending probability starts at $1$ and gradually decreases. The length of the growing phase is the same
in every iteration, while the length of the shrinking phase doubles
from one iteration to the next.

Each packet divides its local execution into \defn{chunks} of
$m=\lceil d\sqrt{\CollCost}\ln\CollCost\rceil$ consecutive slots,
where $d>0$ is a sufficiently large constant. A packet uses the same
sending probability throughout a chunk, making a fresh independent
transmission decision in each slot. Packets activated at different
times need not have aligned chunks. All phase lengths below are the
full prescribed lengths; the execution stops at the first success.
\smallskip

\noindent{\underline{\it Growing phase.}} For every $k\ge0$, the growing
phase $G_k$ consists of $\lfloor\lg w_0\rfloor+1
=\lfloor\CollCost^\tunableparam\rfloor+1$ chunks. In chunk $i$, for
$i=0,1,\ldots,\lfloor\lg w_0\rfloor$, each packet sends independently
in every slot with probability $2^i/w_0$
(Lines~\ref{alg:second-line-dynamic}--\ref{alg:dynamic-second-sending-probability}).
It doubles after each chunk
(Line~\ref{alg:dynamic-havle-assign}), with final value
$2^{\lfloor\lg w_0\rfloor}/w_0\in(1/2,1]$. Thus, every growing phase
has length:
$$
|G_k|=m\left(\lfloor\lg w_0\rfloor+1\right)
=\Theta\left(\CollCost^{1/2+\tunableparam}\ln\CollCost\right).
$$

\noindent{\underline{\it Shrinking phase.}} For every $k\ge0$, the
shrinking phase $S_k$ consists of $2^k$ chunks. In chunk $j$, for
$j=0,1,\ldots,2^k-1$, each packet sends independently in every slot
with probability $2^{-j}$
(Lines~\ref{alg:while_second_dynamic}--\ref{alg:dynamic-third-sending-probability}).
Thus, the sending probability begins at $1$ and is halved after every
chunk. The length of $S_k$ is:
$$
|S_k|=m2^k=\Theta\left(2^k\sqrt{\CollCost}\ln\CollCost\right).
$$

A packet proceeds through iterations $(G_k,S_k)$, for
$k=0,1,2,\ldots$, until a successful transmission occurs.
\medskip

\noindent{\bf Batches and Contention.} Packets activated at the same
time form a \defn{batch}. Let $b\le n$ be the number of batches
activated during the execution, and let $n_i$ be the size of batch
$i$. Thus, $\sum_{i=1}^{b}n_i\le n$; some packets may still be
inactive at the first success. Packets in the same batch follow the
phase sequence $G_0,S_0,G_1,S_1,\ldots$ in lockstep, beginning at the
batch's activation time.

For any reached global slot $t$, let $p_u(t)$ denote the sending
probability of an active packet $u$. The \defn{contention} in slot
$t$ is $\Con(t):=\sum_u p_u(t)$, where the sum is over the active
packets. Let $f_i(t)$ denote the contention contributed by batch $i$.
We index batches in order of activation, so $f_1$ is the contribution
of the first batch activated. Then $\Con(t)=\sum_{i=1}^{b}f_i(t)$,
where a batch not yet activated contributes $0$. We say that $f_i$
is in a growing or shrinking phase when its packets execute that
phase of \mainAlgorithmAcrynom.

\begin{definition}\label{def:contention-regions}
Fix constants $\alpha\ge1$ and $\beta\ge8\alpha$, and assume
$\CollCost$ is sufficiently large that
$\beta/\sqrt{\CollCost}\le1/2$. The \emph{good contention region} is:
$$
\mathbb{G}=\left[
\frac{\alpha}{\sqrt{\CollCost}},
\frac{\beta}{\sqrt{\CollCost}}
\right].
$$
A slot $t$ has \emph{low contention} if
$\Con(t)<\alpha/\sqrt{\CollCost}$, \emph{good contention} if
$\Con(t)\in\mathbb{G}$, and \emph{high contention} if
$\Con(t)>\beta/\sqrt{\CollCost}$. Every reached slot falls into
exactly one of these categories.
\end{definition}

%%%%%%%%%%%%%%%%%%%%%%%%%%%%%%%%%%%%%%%%%%%

\subsection{Context and Overview}\label{sec:tech-overview}

We now provide intuition for the design of \mainAlgorithmAcrynom and
explain how its growing and shrinking phases control both latency and
collision cost in the dynamic setting.

\subsubsection{\bf A Baseline via Backoff.}\label{s:baseline}
A natural approach to symmetry breaking is to ``back off'' by gradually
decreasing each packet's sending probability. Classical dynamic-wakeup
algorithms obtain strong latency guarantees using such schedules
\cite{jurdzinski2002probabilistic,de2022contention}. For example, the
\textsc{DecreaseSlowly} protocol analyzed by De Marco et
al.~\cite{de2022contention} begins with transmission probability $1/2$.
If two packets begin together, the first transmission round already has
a constant probability of collision and therefore incurs
$\Omega(\CollCost)$ expected collision cost.

Thus, a key obstacle is setting the initial sending probability
suitably low. Since $n$ is unknown, however, it is not clear how low
that probability should be.

\subsubsection{\bf Design Choices and Overview.}
\label{sec:design-choices}
We now explain the choice of initial probability, the role of the
shrinking phases, and why success during the first batch's initial
growing phase is very likely in the large-$\CollCost$ case of
Theorem~\ref{thm:dynamic}.
\medskip

\noindent\textbf{Aiming for Good Contention.}
Suppose a slot has contention $\Con(t)=\Theta(1/\sqrt{\CollCost})$.
Its conditional success probability is
$\Theta(1/\sqrt{\CollCost})$, whereas its collision probability is
$O(1/\CollCost)$. Thus, the slot's conditional expected contribution
to collision cost is only $O(1)$.
More generally, Lemma~\ref{lem:success-good-contention} shows that
encountering
$m=\Theta(\sqrt{\CollCost}\ln\CollCost)$ good-contention slots without
a success has probability at most $\CollCost^{-\alpha d/e}$. This
motivates both the good contention region and the chunk length.
\smallskip

\noindent\textbf{Setting the Initial Sending Probability.}
For one batch of $n$ packets sending with probability $1/w$, the
contention is $n/w$. We choose
$w_0=2^{\CollCost^{\tunableparam}}$, giving the two cases in
Theorem~\ref{thm:dynamic}:
\begin{enumerate}[leftmargin=12pt]
\item \textbf{Large-$\CollCost$ regime:}
$\CollCost\ge K\lg^{1/\tunableparam}n$.
For sufficiently large $K$, this implies
$n/w_0\le1/\sqrt{\CollCost}$
(Lemma~\ref{lem:first-growing-survival}). Thus, both the initial
contention and the contribution of newly activated packets while they
remain in their first chunks are small.

\item \textbf{Small-$\CollCost$ regime:}
$\CollCost<K\lg^{1/\tunableparam}n$.
Here $\CollCost=O(\lg^{1/\tunableparam}n)$. The general latency bound
below, together with the fact that each slot contributes at most
$\CollCost$ to collision cost, gives
$O(n\log^{\Theta(1/\tunableparam)}n)$ for both expectations.
\end{enumerate}
These cases are only for the analysis; \mainAlgorithmAcrynom executes
the same schedule in both.
\medskip

\noindent\textbf{Controlling Long Executions.}
In the dynamic setting, growing phases of other batches can interrupt
the decrease in contention during the first batch's shrinking phases.
However, growing phases have fixed length while shrinking phases become
progressively longer. Once a shrinking phase is sufficiently long, the
execution must contain a long interval in which every active batch is
shrinking. Lemmas~\ref{lem:few-growing}--\ref{lem:roi-visits} show that
each such completed phase contains $m$ consecutive good-contention
slots.

Consequently, completing $j$ sufficiently late shrinking phases without
a success requires surviving at least $jm$ good-contention slots. The
probability of this event decreases geometrically in $j$, yielding:
$$
O\left(
n\CollCost^{1/2+2\tunableparam}\ln\CollCost\,\lg^3 n
\right)
$$
expected remaining latency
(Lemma~\ref{lem:latency-continuation}). This bound applies in both cases
of the main theorem: it gives the small-$\CollCost$ guarantee and
controls rare long executions in the large-$\CollCost$ case.
\medskip

\noindent\textbf{Success During the First Growing Phase.}
Let $L$ denote the latency.  Suppose $\CollCost\ge K\lg^{1/\tunableparam}n$, and let $g$ be the
length of the first batch's initial growing phase. Through slot $g$,
every activated packet remains in its own initial growing phase. If no
success occurs, contention must eventually rise from at most
$1/\sqrt{\CollCost}$ to at least $1/2$. The first such crossing is
preceded by $m$ slots with contention in $[1/8,1/2)$.

Each of these slots has a constant conditional success probability.
Therefore, we can show (Lemma~\ref{lem:first-growing-survival}):
$$
\Pr(L>g)
\le
\exp\left(-\frac{m}{8e}\right).
$$
We initially charge
$g=O(\CollCost^{1/2+\tunableparam}\ln\CollCost)$ slots, while any
remaining duration is incurred only on this exponentially unlikely
event. Combining this probability with
Lemma~\ref{lem:latency-continuation} shows that executions continuing
beyond $g$ contribute only $O(1)$ additional slots in expectation.
Hence, we have (Lemma~\ref{lem:latency-large-cost}):
$$
\mathbb{E}[L]
=
O\left(
\CollCost^{1/2+\tunableparam}\ln\CollCost
\right).
$$
\noindent\textbf{Handling High Contention.}
Lemma~\ref{lem:collision-low-good} shows that low- and
good-contention slots contribute only $O(\sqrt{\CollCost})$ expected
collision cost in total.

For high contention in the large-$\CollCost$ case, we again split the
execution at slot $g$. Reaching high contention by then requires
surviving $m$ good-contention slots, an event of probability at most
$\CollCost^{-\alpha d/e}$. After $g$, the probability that the execution
continues is at most $\exp(-m/(8e))$. These probabilities are small
enough to offset the possible collision cost incurred in the
corresponding executions. Thus, high-contention slots contribute only
$O(1)$ expected collision cost
(Lemma~\ref{lem:collision-high}), giving expected collision cost over the execution
$O(\sqrt{\CollCost})$.
\medskip

\noindent\textbf{Relation to Backoff.}
Our algorithm uses familiar backoff-and-backon probability updates, but
the central challenge is to organize them around both latency and
collision cost. Rather than seek a constant success probability per
slot, we target contention of order $1/\sqrt{\CollCost}$, balancing a
success probability of order $1/\sqrt{\CollCost}$ against a collision
probability of only $O(1/\CollCost)$. In the large-$\CollCost$ regime,
the low initial sending probability keeps costly contention unlikely
during the first batch's initial growing phase, while the increasingly
long shrinking phases control executions that survive that phase. This
gives bounds on both expected latency and expected collision cost.
\medskip

\noindent{\bf Role of {\boldmath{$\tunableparam$}}.}
The parameter $\tunableparam$ is a fixed constant chosen before
execution, with $0<\tunableparam\le1/2$. It controls the initial
window $w_0=2^{\CollCost^\tunableparam}$. Taking a smaller value
shortens each growing phase and improves the large-$\CollCost$
latency bound, but raises the exponent $1/\tunableparam$ in the
threshold $K\lg^{1/\tunableparam}n$ and the polylogarithmic exponent
in the small-$\CollCost$ bound.

The upper limit $\tunableparam\le1/2$ comes from controlling
executions that survive the first batch's initial growing phase.
Their probability is at most $\exp(-m/(8e))$
(Lemma~\ref{lem:first-growing-survival}), whereas the bound on
expected remaining latency contains a factor of $n$
(Lemma~\ref{lem:latency-continuation}). In the large-$\CollCost$
regime, our choice of $K$ gives
$n\le2^{\CollCost^\tunableparam/2}$. The survival probability must
decrease fast enough to offset this factor and the remaining
polynomial and logarithmic factors. Since
$m=\Theta(\sqrt{\CollCost}\ln\CollCost)$, the restriction
$\tunableparam\le1/2$ ensures
$\CollCost^\tunableparam\le\sqrt{\CollCost}$, so the exponential
decay dominates. This comparison is used in
Lemmas~\ref{lem:latency-large-cost} and~\ref{lem:collision-high}.
For any fixed $\tunableparam>1/2$, however,
$\CollCost^\tunableparam$ eventually grows faster than
$\sqrt{\CollCost}\ln\CollCost$, so this comparison no longer gives
the required bound.

{\fontsize{9.5}{12}\selectfont
\begin{algorithm}[t!]
\caption{\bf \mainAlgorithm}
\label{our-dynamic-algoritm-LE}

\BlankLine
\tcp*[h]{Initialization}\\
$w_{0} \leftarrow 2^{\CollCost^{\tunableparam}}$
\label{alg:first-line-dynamic}\\
$m \leftarrow \lceil d\sqrt{\CollCost}\ln{\CollCost} \rceil$
\label{alg:assign-m}\\
$k \leftarrow 0$
\label{alg:assign-k}\\

\medskip
\medskip

\tcp*[h]{Iterations}\\
\While{true}{
\medskip

\tcp*[h]{Growing Phase $G_k$}\\
$\currWin \leftarrow w_{0}$
\label{alg:second-line-dynamic}\\

\While{$\currWin \geq 1$}{
    \label{alg:dynamic-first-while-loop}

    \For{slot $i=1$ \KwTo $m$}{
        \label{alg:dynamic-first-while-for-loop}
        \vspace{-2pt}
        Send packet with probability $1/\currWin$
        \label{alg:dynamic-second-sending-probability}\\

        \If{success}{
            Terminate
        }
    }

    $\currWin \leftarrow \currWin/2$
    \label{alg:dynamic-havle-assign}\\
}

\medskip
\medskip

\tcp*[h]{Shrinking Phase $S_k$}\\
$\currWin \leftarrow 1$\\

\For{$s=0$ \KwTo $2^k-1$}{
    \label{alg:while_second_dynamic}

    \For{slot $i=1$ \KwTo $m$}{
        \label{alg:dynamic-second-while-for-loop}
        \vspace{-2pt}
        Send packet with probability $1/\currWin$
        \label{alg:dynamic-third-sending-probability}\\

        \If{success}{
            Terminate
        }
    }

    $\currWin \leftarrow 2\currWin$
    \label{alg:dynamic-double-assign}\\
}

$k \leftarrow k+1$
}

\end{algorithm}
}

\section{Upper-Bound Analysis}\label{sec:dynamic}

In this section, we analyze the performance of \mainAlgorithmAcrynom for the dynamic wakeup problem.
\ifshowproofs\else
Due to space constraints, proofs of several  lemmas are deferred to the appendix, which contains the full version of the paper.
\fi

\subsection{Preliminaries and Helper Lemmas}

For our analysis to handle an adaptive adversary, we need to define the notion of a history. For each slot $s$, let $A_s$ denote the set of packets activated in
slot $s$, let $T_s$ denote the set of packets that transmit in that
slot, and let $Y_s$ denote the resulting channel outcome.  A
\defn{pre-transmission history for slot} {\boldmath{$t$}} is a tuple:
$$h_t=(A_1,T_1,Y_1,\ldots,A_{t-1},T_{t-1},Y_{t-1},A_t).$$
Thus, {\boldmath{$h_t$}} records the entire execution through the end of slot
$t-1$, together with the adversary's activation decision for slot
$t$, but not the transmission decisions in slot $t$.  For any reachable history $h_t$, the set of active packets and their
sending probabilities in slot $t$ are fixed.

\begin{lemma}\label{lem:bender} (Bender et al.~\cite{bender:fully})
Fix any reachable pre-transmission history $h_t$. Suppose that,
conditioned on $h_t$, each active packet transmits independently in
slot $t$ with probability at most $1/2$. Then:
$$
\Pr(\text{success in slot }t\mid h_t)
\ge
\frac{\Con(t)}{e^{2\Con(t)}}.
$$
\end{lemma}

\begin{lemma}\label{lem:success-good-contention}
For every integer $r\ge 1$, the probability that \mainAlgorithmAcrynom encounters at
least $r$ good-contention slots without a successful transmission
in any of the first $r$ such slots is at most:
$$
\left(1-\frac{\alpha}{e\sqrt{\CollCost}}\right)^r.
$$
\end{lemma}

\ifshowproofs
\begin{proof}
Fix any reachable pre-transmission history $h_t$ under which
slot $t$ has good contention. The active packets and their sending
probabilities are fixed by $h_t$, and their transmission decisions
in slot $t$ are independent. By Definition~\ref{def:contention-regions}:
$$
\frac{\alpha}{\sqrt{\CollCost}}
\le\Con(t)\le\frac{\beta}{\sqrt{\CollCost}}\le\frac12.
$$
Every sending probability is at most $\Con(t)$, so
Lemma~\ref{lem:bender} gives:
$$
\Pr(\text{success in slot }t\mid h_t)
\ge\frac{\Con(t)}{e^{2\Con(t)}}
\ge\frac{\alpha}{e\sqrt{\CollCost}}.
$$
Let $p=\alpha/(e\sqrt{\CollCost})$. For $j\ge1$, let $E_j$ be
the event that the execution encounters at least $j$
good-contention slots and remains unsuccessful through the end
of the $j$-th. Let $E_0$ be the certain event.

Reaching the $j$-th good-contention slot requires $E_{j-1}$ to
occur. For every pre-transmission history leading to that slot,
its conditional failure probability is at most $1-p$.
Averaging over those histories gives:
$$
\Pr(E_j)\le(1-p)\Pr(E_{j-1}).
$$
If another success occurs first, or no further good-contention
slot is reached, then $E_j$ does not occur. Iterating the bound
gives, for every integer $r\ge1$:
$$
\Pr(E_r)\le(1-p)^r
=\left(1-\frac{\alpha}{e\sqrt{\CollCost}}\right)^r
$$
as claimed.
\end{proof}
\fi

Our collision-cost analysis uses the following fact. For a
fixed total probability $\sum_i p_i$, the sum of all products of any
fixed number $r$ of distinct probabilities is maximized when the
probabilities are equal.

\begin{lemma}\label{lem:equal-probabilities}
Fix an integer $r\in\{2,\ldots,N\}$ and a value
$\lambda\in[0,N]$. Among all choices of probabilities
$p_1,\ldots,p_N\in[0,1]$ satisfying $\sum_{i=1}^{N} p_i=\lambda,
$, the quantity:
$$
\sum_{\substack{S\subseteq[N]\\ |S|=r}}
\prod_{i\in S}p_i
$$
is maximized when
$p_1=p_2=\cdots=p_N={\lambda/N}$.
\end{lemma}

\ifshowproofs
\begin{proof}
Consider any two probabilities $p_i$ and $p_j$, while holding all
other probabilities fixed. For ease of presentation, define:
$$
R_2=
\sum_{\substack{T\subseteq[N]\setminus\{i,j\}\\ |T|=r-2}}
\prod_{\ell\in T}p_\ell,
$$
$$
R_1=
\sum_{\substack{T\subseteq[N]\setminus\{i,j\}\\ |T|=r-1}}
\prod_{\ell\in T}p_\ell,
$$
and
$$
R_0=
\sum_{\substack{T\subseteq[N]\setminus\{i,j\}\\ |T|=r}}
\prod_{\ell\in T}p_\ell.
$$
Then:
$$
\sum_{\substack{S\subseteq[N]\\ |S|=r}}
\prod_{\ell\in S}p_\ell
=
p_ip_jR_2+(p_i+p_j)R_1+R_0.
$$
Replace $p_i$ and $p_j$ by their average:
$$
p_i'=p_j'=\frac{p_i+p_j}{2}.
$$
Since $p_i'+p_j'=p_i+p_j$, this replacement preserves the constraint
$\sum_{\ell=1}^{N}p_\ell=\lambda$. It also leaves the term
$(p_i+p_j)R_1$ unchanged. Thus, the only term that can change is
$p_ip_jR_2$. Moreover:
$$
p_i'p_j'
=
\left(\frac{p_i+p_j}{2}\right)^2
=
p_ip_j+\frac{(p_i-p_j)^2}{4}
\ge p_ip_j.
$$
Since $R_2\ge0$, averaging $p_i$ and $p_j$ does not decrease the
quantity being maximized.

Now repeatedly apply this operation to a largest and a smallest
probability among $p_1,\ldots,p_N$. Let $\mu=\lambda/N$. Each
averaging step decreases:
$$
\sum_{i=1}^{N}(p_i-\mu)^2
$$
by exactly:
$$
\frac{(p_i-p_j)^2}{2},
$$
where $p_i$ and $p_j$ are the two probabilities being averaged.
Since this sum is nonnegative, the difference between the largest
and smallest probabilities tends to zero. Their average remains
$\mu$, so:
$$
p_1,\ldots,p_N
\longrightarrow
\mu=\frac{\lambda}{N}.
$$
The quantity being maximized never decreases during this process.
Since it is a continuous function of $p_1,\ldots,p_N$, its value
therefore converges to its value at
$p_1=\cdots=p_N=\lambda/N$. Hence, the equal choice has value at
least as large as the original choice and therefore maximizes the
quantity.
\end{proof}
\fi

%%%%%%%%%%%%%%%%%%%%%%%%%%%%%%%%%%%%%%%%%%%%%%%%%%%%%%%%%%%%%%%%%%%%%%%%%
%%%%%%%%%%%%%%%%%%%%%%%%%%%%%%%%%%%%%%%%%%%%%%%%%%%%%%%%%%%%%%%%%%%%%%%%%%
\begin{lemma}\label{lem:collision-probability}
Fix any reachable pre-transmission history under which slot $t$
has low or good contention. The conditional collision probability
$P_{\mathrm{col}}(t)$ satisfies
$P_{\mathrm{col}}(t)=O\left(\Con(t)^2\right)=O(1/\CollCost)$.
\end{lemma}

\ifshowproofs
\begin{proof}
Condition on the stated history, and let $\lambda=\Con(t)$.
Suppose $N$ packets are active, with sending probabilities
$p_1,\ldots,p_N$, so $\sum_{i=1}^{N}p_i=\lambda$. These
probabilities are fixed, and the packets transmit independently.
All probabilities below are conditional on this history. If $N<2$,
no collision can occur, so assume $N\ge2$.

Let $X_t$ denote the number of packets that transmit in slot $t$.
For any $r\in\{2,\ldots,N\}$, note that:
\begin{align*}
\Pr(X_t=r)
&=
\sum_{\substack{S\subseteq[N]\\ |S|=r}}
\left(\prod_{i\in S}p_i\right)
\left(\prod_{j\notin S}(1-p_j)\right)\\
&\le
\sum_{\substack{S\subseteq[N]\\ |S|=r}}
\prod_{i\in S}p_i\\
&\le
\binom{N}{r}
\left(\frac{\lambda}{N}\right)^r\\
&\le
\frac{\lambda^r}{r!}.
\end{align*}
The second-to-last line follows from
Lemma~\ref{lem:equal-probabilities}, since the right-hand side is
maximized when all sending probabilities are equal to $\lambda/N$.
The last line follows from
$\binom{N}{r}\le N^r/r!$.

A collision occurs when $X_t\ge2$. So:
\begin{align*}
P_{\mathrm{col}}(t)
&=
\Pr(X_t\ge2)\\
&=
\sum_{r=2}^{N}\Pr(X_t=r)\\
&\le
\sum_{r=2}^{N}\frac{\lambda^r}{r!}\\
&\le
e^\lambda-1-\lambda,
\end{align*}
where the last line follows from the Taylor expansion:
$$
e^\lambda
=
\sum_{r=0}^{\infty}\frac{\lambda^r}{r!}
=
1+\lambda+\sum_{r=2}^{\infty}\frac{\lambda^r}{r!}.
$$
Since $t$ has low or good contention, we have
$\lambda=\Con(t)\le\beta/\sqrt{\CollCost}$. For sufficiently large
$\CollCost$, this implies $\lambda\le1/2$. By Taylor expansion:
$$
e^\lambda-1-\lambda
=
\sum_{r=2}^{\infty}\frac{\lambda^r}{r!}
=
\lambda^2\sum_{j=0}^{\infty}\frac{\lambda^j}{(j+2)!}.
$$
For $\lambda\le1/2$, the remaining series is bounded by a constant,
and hence:
$$
e^\lambda-1-\lambda
=
O(\lambda^2).
$$
Consequently:
$$
P_{\mathrm{col}}(t)
=
O\left(\lambda^2\right)
=
O\left(\Con(t)^2\right),
$$
and since $\Con(t)\le\beta/\sqrt{\CollCost}$, we obtain
$P_{\mathrm{col}}(t)=O(1/\CollCost)$.
\end{proof}
\fi

%%%%%%%%%%%%%%%%%%%%%%%%%%%%%%%%%
%%%%%%%%%%%%%%%%%%%%%%%%%%%%%%%%%
%%%%%%%%%%%%%%%%%%%%%%%%%%%%%%%%%

\subsection{Latency Analysis}

We first use the shrinking phases to bound the expected remaining
latency from any history reached during or immediately after the
first batch's initial growing phase. When
$\CollCost\ge K\lg^{1/\tunableparam}n$, the probability of completing
that phase without a success is at most $\exp(-m/(8e))$.
Combining these bounds gives the sharper expected-latency guarantee
for the large-$\CollCost$ regime.

Let us number the slots from the first activation (thus, we can ignore any prefix of slots where no packets are active). Recall that phase
lengths refer to their full prescribed lengths, without truncation at
a success. Let {\boldmath{$g$}} denote the length of a growing phase. Since each growing phase consists of $\lfloor\lg w_0\rfloor+1$ chunks, each of length $m$, we have $g=m\left(\lfloor\lg w_0\rfloor+1\right)
=
\Theta(
\CollCost^{1/2+\tunableparam}\ln\CollCost
)$

For each $k\ge0$, let {\boldmath{$S_k$}} denote the interval prescribed
for the first batch's $k$-th shrinking phase, so $|S_k|=m2^k$.
Any statement concerning all of $S_k$ applies only when the execution
reaches and completes that phase.

\begin{definition}\label{d:shrinking-phase}
An interval of consecutive slots $J=[\tau,\tau+\ell)\subseteq S_k$
is a \defn{clean shrinking window} if no active batch is in a
growing phase in any slot of $J$. Such a window is \defn{maximal}
if it cannot be extended within $S_k$ while retaining this property.
\end{definition}

\noindent Set:
$$
D=\left\lceil\lg\left({n\sqrt{\CollCost}}/{\beta}\right)\right\rceil+1.
$$
Thus, {\boldmath{$D$}}$=O(\lg n+\lg\CollCost)$.

%%%%%%%%%%%%%%%%%%%%%%%%%%
%%%%%%%%%%%%%%%%%%%%%%%%%%

\begin{lemma}\label{lem:few-growing}
Let $I$ be any interval of $\ell$ consecutive slots that the execution
reaches, with $\ell\ge2g$. For any batch, at most $\lg \ell+O(1)$ of its
growing phases intersect $I$.
\end{lemma}

\ifshowproofs
\begin{proof}
Suppose $r$ growing phases of the batch intersect $I$. The claim is
immediate if $r\le1$. Otherwise, between these phases lie $r-1$
consecutive shrinking phases. Their total length is at least:
$$
m\left(1+2+\cdots+2^{r-2}\right)
=
m\left(2^{r-1}-1\right).
$$
The first intersecting growing phase can begin before $I$, but by
fewer than $g$ slots. Similarly, the last can end fewer than $g$
slots after $I$. Consequently:
$$
m\left(2^{r-1}-1\right)
\le
\ell+2g
\le
2\ell.
$$
It follows that:
$$
r
\le
1+\lg(1+2\ell/m)
\le
\lg \ell+O(1),
$$
where the last step uses $m\ge1$ and $\ell\ge1$.
\end{proof}
\fi

%%%%%%%%%%%%%%%%%%%%%%%%%%%%%%%%%%%%%%%%%%
%%%%%%%%%%%%%%%%%%%%%%%%%%%%%%%%%%%%%%%%%%

\begin{lemma}\label{lem:window-forces-roi}
Let $J=[\tau,\tau+\ell)$ be a clean shrinking window with
$\ell\ge(D+2)m$ and $1\le\Con(\tau)\le n$. Then $J$ contains
$m$ consecutive slots of good contention.
\end{lemma}

\ifshowproofs
\begin{proof}
No new batch can be activated during $J$, since it would begin in a
growing phase. All active batches remain in shrinking phases, so
$\Con(t)$ is nonincreasing throughout $J$.

Since $\beta/\sqrt{\CollCost}\le1/2$ (recall Definition \ref{def:contention-regions}), we have 
$\Con(\tau)>\beta/\sqrt{\CollCost}$. Between slots $\tau$ and
$\tau+(D+1)m$, each active batch crosses at least $D$ of its own
chunk boundaries. Each such boundary halves its contribution, so:
$$
\Con\left(\tau+(D+1)m\right)
\le\frac{\Con(\tau)}{2^D}
\le\frac{n}{2^D}
<\frac{\beta}{\sqrt{\CollCost}}.
$$
Let $t$ be the first slot in this interval with
$\Con(t)\le\beta/\sqrt{\CollCost}$. The preceding slot has
$\Con(t-1)>\beta/\sqrt{\CollCost}$. Between two consecutive slots
of $J$, each batch's contribution either stays unchanged or halves.
Therefore, we have that 
$\Con(t)>{\beta}/({2\sqrt{\CollCost}})$.
During the $m$ consecutive slots beginning at $t$, each batch
crosses at most one further chunk boundary. Thus, for every
$s\in[t,t+m)$:
$$
\Con(s)\ge\frac{\Con(t)}2
>\frac{\beta}{4\sqrt{\CollCost}}
\ge\frac{\alpha}{\sqrt{\CollCost}},
$$
where the last inequality uses $\beta\ge8\alpha$. Since every active batch remains in a shrinking phase throughout $J$
and no new batch is activated there, the aggregate contention is
nonincreasing on $J$. Hence, $\Con(s)\le\Con(t)\le{\beta}/{\sqrt{\CollCost}}$. These $m$ slots lie in $J$, since $t\le\tau+(D+1)m$ and
$\ell\ge(D+2)m$, and hence all have good contention.
\end{proof}
\fi

%%%%%%%%%%%%%%%%%%%%%%%%%%%%%%%%%%%%%%%%%%
%%%%%%%%%%%%%%%%%%%%%%%%%%%%%%%%%%%%%%%%%%

\begin{lemma}\label{lem:long-good}
Let $k^*=\lceil\lg n+\tunableparam\lg\CollCost+3\lg\lg n+c\rceil$,
where $c>0$ is a sufficiently large constant. For every $k\ge k^*$,
if the execution reaches the end of $S_k$, then $S_k$ contains a
maximal clean shrinking window of length at least $(D+2)m$.
\end{lemma}

\ifshowproofs
\begin{proof}
Let $\ell=|S_k|=m2^k$. Since $g/m\le\CollCost^\tunableparam+1$, the
choice of $k^*$ ensures $\ell\ge2g$. By Lemma~\ref{lem:few-growing},
each batch has at most $k+\lg m+O(1)$ growing phases intersecting
$S_k$. There are at most $n$ batches, so the total number $Q$ of
such phases satisfies $Q=O(n(k+\lg m))$.

These growing phases cover at most $Qg$ slots of $S_k$. Removing
the covered slots leaves at most $Q+1$ maximal clean shrinking
windows. If every such window had length less than $(D+2)m$, then:
$\ell<Qg+(Q+1)(D+2)m$. Dividing by $m$ and using the bounds on $Q$ and $g/m$ would give,
for a fixed constant $c_1>0$:
\begin{equation}\label{eq:contra}
2^k<c_1n(k+\lg m)(\CollCost^\tunableparam+D+2).
\end{equation}
At $k=k^*$, the left-hand side satisfies:
$$
2^{k^*}\ge2^c n\CollCost^\tunableparam\lg^3 n.
$$
On the other hand, $\CollCost=O(\mathrm{poly}(n))$ implies
$\lg m=O(\lg n)$, $D=O(\lg n)$, and $k^*=O(\lg n)$. Since
$\CollCost^\tunableparam\ge1$, the right-hand side of
Equation~\eqref{eq:contra} is:
$$
c_1n(k^*+\lg m)(\CollCost^\tunableparam+D+2)
=O\hspace{-2pt}\left(n\CollCost^\tunableparam\lg^2 n\right).
$$
For sufficiently large $c$ and $n$, this is smaller than $2^{k^*}$.
Moreover, $2^k/(k+\lg m)$ is increasing for $k\ge k^*$, so
Equation~\eqref{eq:contra} cannot hold for any $k\ge k^*$.
Therefore, at least one maximal clean shrinking window in $S_k$
has length at least $(D+2)m$.
\end{proof}
\fi

%%%%%%%%%%%%%%%%%%%%%%%%%%%%%%%%%%%%%%%
%%%%%%%%%%%%%%%%%%%%%%%%%%%%%%%%%%%%%%%

\begin{lemma}\label{lem:good-start}
Let $J=[\tau,\tau+\ell)$ be a maximal clean shrinking window
inside a completed phase $S_k$. Then $1\le\Con(\tau)\le n$.
\end{lemma}
\ifshowproofs
\begin{proof}
The upper bound follows trivially because at most $n$ packets are active, each with sending probability at most $1$.

If $J$ begins at the start of $S_k$, then the first batch has just
entered its shrinking phase. Its sending probability is $1$, so
$\Con(\tau)\ge f_1(\tau)=n_1\ge1$.

Otherwise, $\tau-1$ also lies in $S_k$. By maximality of $J$, at
least one batch is in a growing phase at $\tau-1$, whereas every
active batch is in a shrinking phase at $\tau$. Hence some batch
$i$ has just entered a shrinking phase at $\tau$. It contributes
$f_i(\tau)=n_i\ge1$, giving the lower bound.
\end{proof}
\fi

%%%%%%%%%%%%%%%%%%%%%%%%%%%%%%%%%%%%%%%
%%%%%%%%%%%%%%%%%%%%%%%%%%%%%%%%%%%%%%%

\begin{lemma}\label{lem:roi-visits}
Let $k^*=\lceil\lg n+\tunableparam\lg\CollCost+3\lg\lg n+c\rceil$,
where $c>0$ is a sufficiently large constant. For every $k\ge k^*$,
if the first batch completes $S_k$, then that phase contains $m$
consecutive slots of good contention.
\end{lemma}
\ifshowproofs
\begin{proof}
By Lemma~\ref{lem:long-good}, $S_k$ contains a maximal clean
shrinking window $J=[\tau,\tau+\ell)$ with $\ell\ge(D+2)m$.
Lemma~\ref{lem:good-start} gives $1\le\Con(\tau)\le n$.
Lemma~\ref{lem:window-forces-roi} therefore supplies the required
$m$ good-contention slots inside $J\subseteq S_k$.
\end{proof}
\fi

%%%%%%%%%%%%%%%%%%%%%%%%%%%%%%%%%%%%%%%
%%%%%%%%%%%%%%%%%%%%%%%%%%%%%%%%%%%%%%%

For the remaining latency arguments, let {\boldmath{$L$}} denote the
latency, with $L=\infty$ if no success occurs. 
Let {\boldmath{$B_k$}} be the number of slots prescribed by \mainAlgorithmAcrynom for iterations $0$
through $k$ of the first batch, counting every phase at its full
length. Thus, we have:
\begin{equation}\label{eq:Bk}
\begin{aligned}
B_k
&=(k+1)g+m(2^{k+1}-1)\\
&=O\left(2^k\CollCost^{1/2+\tunableparam}\ln\CollCost\right).
\end{aligned}
\end{equation}
For every $k\ge0$ and every execution, the number of slots actually
executed among the first batch's iterations $0,\ldots,k$ is
$\min\{L,B_k\}$. In particular, if $L>B_k$, then all $B_k$ prescribed
slots of these iterations are completed without a success.

\begin{lemma}\label{lem:latency-continuation}
For every reachable pre-transmission history $h_s$ with
$1\le s\le g+1$, the expected number of remaining slots, including
slot $s$, satisfies $
\mathbb{E}[L-s+1\mid h_s]
=
O(
n\CollCost^{1/2+2\tunableparam}
\ln\CollCost\,\lg^3 n
)$.
\end{lemma}

\begin{proof}
Fix such a history $h_s$. All shrinking phases
$S_{k^*},S_{k^*+1},\ldots$ begin at or after $s$.
For $j\ge1$, the event $L>B_{k^*+j-1}$ requires completing the
first $j$ of these phases without a success.
By Lemma~\ref{lem:roi-visits}, these disjoint phases contain at
least $jm$ good-contention slots.

At every reached good-contention slot $t$, conditioned on any
pre-transmission history $h_t$, the active packets transmit
independently with fixed probabilities. Each probability is at most
$\Con(t)\le\beta/\sqrt{\CollCost}\le1/2$.
Lemma~\ref{lem:bender} therefore gives:
$$
\Pr(\text{success in slot }t\mid h_t)
\ge\frac{\Con(t)}{e^{2\Con(t)}}
\ge\frac{\alpha}{e\sqrt{\CollCost}}.
$$

By Lemma~\ref{lem:roi-visits}, completing these $j$ shrinking
phases without a success requires remaining unsuccessful through at
least $jm$ good-contention slots. Conditioned on the pre-transmission
history before any such slot, Lemma~\ref{lem:bender} gives a success
probability of at least $\alpha/(e\sqrt{\CollCost})$. Thus, each time
another good-contention slot is reached, its conditional probability
of failure is at most
$1-\alpha/(e\sqrt{\CollCost})$.

Define the probability $q=\CollCost^{-\alpha d/e}$, where recall $d$ is the chunk-size constant. Conditioning successively on the history before each of the first
$jm$ good-contention slots therefore gives:
$$
\begin{aligned}
\Pr\left(L>B_{k^*+j-1}\mid h_s\right)\hspace{-2pt}
\le
\left(1-\frac{\alpha}{e\sqrt{\CollCost}}\right)^{jm}
\hspace{-9pt}\le
\exp\left(-\frac{\alpha jm}{e\sqrt{\CollCost}}\right)
\le
\CollCost^{-j\alpha d/e}
\hspace{-2pt}=
q^j,
\end{aligned}
$$
where the second inequality uses $1-x\le e^{-x}$ and the third
uses $m\ge d\sqrt{\CollCost}\ln\CollCost$. Grouping the remaining slots by the prescribed endpoints gives:
\begin{align}
\mathbb{E}[L-s+1\mid h_s]
&\le B_{k^*}-s+1
+\sum_{j\ge1}
\left(B_{k^*+j}-B_{k^*+j-1}\right)q^j
\notag\\
&\le B_{k^*}+\sum_{j\ge1}q^j B_{k^*+j}.
\label{eq:latency-sum}
\end{align}
By the definition of $k^*$, $2^{k^*}=O(n\CollCost^\tunableparam\lg^3 n)$.
Thus, Equation~\eqref{eq:Bk} gives, for every $j\ge0$:
$$
B_{k^*+j}
=O\left(2^j n\CollCost^{1/2+2\tunableparam}
\ln\CollCost\,\lg^3 n\right).
$$
Choose $d$ sufficiently large that $q\le1/4$. Then
$\sum_{j\ge0}(2q)^j\le2$, so Equation~\eqref{eq:latency-sum}
implies $\mathbb{E}[L-s+1\mid h_s]
=O\left(n\CollCost^{1/2+2\tunableparam}\ln\CollCost\,\lg^3 n\right)$.
\end{proof}

The preceding lemma controls the rare executions that continue for a
long time. We now show that, in the large-$\CollCost$ regime,
reaching the end of the first batch's initial growing phase without
a success is exponentially unlikely.

\begin{lemma}\label{lem:first-growing-survival}
Assume $\CollCost\ge K(\lg^{1/\tunableparam} n)$, where $K>0$
is a sufficiently large constant depending only on $\tunableparam$.
The probability that \mainAlgorithmAcrynom completes the first batch's initial
growing phase without a success satisfies $\Pr(L>g)\le\exp\left(-{m}/{(8e)}\right)$.
\end{lemma}

\begin{proof}
We first bound the contribution of newly activated packets.
Choose $K$ large enough that $K^\tunableparam\ge2$ and
$\lg\CollCost\le\CollCost^\tunableparam$ whenever $\CollCost\ge K$; such a choice is possible because $\tunableparam>0$ is a fixed constant. Since $\CollCost\ge K(\lg^{1/\tunableparam} n)$ and
$n\ge2$, we have $\lg n\ge1$ and hence $\CollCost\ge K$.
Moreover, we have:
$$
\CollCost^\tunableparam
\ge K^\tunableparam\lg n
\ge 2\lg n,
$$
where the last inequality uses our choice $K^\tunableparam\ge2$; thus, $\lg n\le\CollCost^\tunableparam/2$. Using this inequality and recalling that
$w_0=2^{\CollCost^\tunableparam}$, we have:
\begin{align}
\frac{n}{w_0}
=
2^{\lg n-\CollCost^\tunableparam}
\le
2^{-\CollCost^\tunableparam/2}
\le
2^{-(\lg\CollCost)/2}
=
\frac{1}{\sqrt{\CollCost}}.\label{eq:1/sqrtC}
\end{align}

Suppose $L>g$. Through slot $g$, every activated packet is still
in its initial growing phase. The first batch's sending probability
in its final chunk is
$2^{\lfloor\lg w_0\rfloor}/w_0>1/2$.
Thus, there is a first slot $\tau\le g$ with $\Con(\tau)\ge1/2$.
During the first $m$ slots, every activated packet is still in its
initial chunk, so:
$$
\Con(t)\le\frac{n}{w_0}\le\frac{1}{\sqrt{\CollCost}}<\frac12.
$$
Hence, it must be the case that $\tau>m$.

Fix $t\in[\tau-m,\tau)$. Since successive chunk boundaries of
a packet are $m$ slots apart, a packet already active at $t$
doubles its sending probability at most once between slots $t$
and $\tau$. A packet activated after $t$ is still in its first
chunk at $\tau$. Consequently:
$$
\frac12
\le
\Con(\tau)
\le
2\Con(t)+\frac{n}{w_0}
\le
2\Con(t)+\frac{1}{\sqrt{\CollCost}},
$$
where the last inequality follows from Equation \ref{eq:1/sqrtC}. For sufficiently large $\CollCost$, this implies:
$$
\Con(t)
\ge
\frac14-\frac{1}{2\sqrt{\CollCost}}
\ge
\frac18.
$$
Also, $\Con(t)<1/2$ by the choice of $\tau$. Thus, $L>g$
requires remaining unsuccessful through at least $m$ slots
with contention in $[1/8,1/2)$. 
For any slot $t$ with contention in $[1/8,1/2)$, conditioned on
the pre-transmission history $h_t$, every active packet has sending
probability at most $\Con(t)<1/2$. Lemma~\ref{lem:bender} therefore
gives:
$$
\Pr(\text{success in slot }t\mid h_t)
\ge
\frac{\Con(t)}{e^{2\Con(t)}}
\ge
\frac{1}{8e}.
$$
Thus, whenever the execution reaches another slot in this contention
range, conditioned on everything that has happened beforehand, the
probability of remaining unsuccessful in that slot is at most
$1-1/(8e)$. Applying this bound successively to $m$ such slots gives:
$$
\Pr(\text{remain unsuccessful through $m$ such slots})
\le
\left(1-\frac{1}{8e}\right)^m.
$$
Since $L>g$ implies that the execution remains unsuccessful through
at least $m$ such slots:
$$
\Pr(L>g)
\le
\left(1-\frac{1}{8e}\right)^m
\le
\exp\left(-\frac{m}{8e}\right),
$$
where the final inequality uses $1-x\le e^{-x}$.
\end{proof}

\begin{lemma}\label{lem:latency-large-cost}
Assume $\CollCost\ge K(\lg^{1/\tunableparam} n)$, where $K>0$
is a sufficiently large constant depending only on $\tunableparam$.
The expected latency of \mainAlgorithmAcrynom is
$O\left(\CollCost^{1/2+\tunableparam}\ln\CollCost\right)$.
\end{lemma}

\begin{proof}
Choose $K$ large enough to satisfy the requirements of
Lemma~\ref{lem:first-growing-survival} and
$K^\tunableparam\ge2$. Executions ending by slot $g$ use, of course, at most $g$ slots. For every
pre-transmission history $h_{g+1}$ reached after surviving the
first growing phase, Lemma~\ref{lem:latency-continuation} gives:
$$
\mathbb{E}[L-g\mid h_{g+1}]
=
O\left(
n\CollCost^{1/2+2\tunableparam}
\ln\CollCost\,\lg^3 n
\right).
$$
Averaging over these histories and applying
Lemma~\ref{lem:first-growing-survival}, we obtain:
\begin{align}
\mathbb{E}[L]
&\le
g+\Pr(L>g)\,
O\left(
n\CollCost^{1/2+2\tunableparam}
\ln\CollCost\,\lg^3 n
\right)\nonumber\\
&\le
g+\exp\left(-\frac{m}{8e}\right)
O\left(
n\CollCost^{1/2+2\tunableparam}
\ln\CollCost\,\lg^3 n
\right)\label{eqn:exp-L}.
\end{align}
Since
$\CollCost\ge K(\lg^{1/\tunableparam} n)$,
we have:
$$
\CollCost^\tunableparam
\ge
K^\tunableparam\lg n
\ge
2\lg n.
$$
This implies that $\lg n\le \frac{\CollCost^\tunableparam}{2}$, which in turn implies $n
\le
2^{\CollCost^\tunableparam/2} = \exp\left(
\frac{\ln 2}{2}\CollCost^{\tunableparam}
\right)$. Similarly, observe that:
$$
\lg^3 n
\le
\frac{1}{8}\CollCost^{3\tunableparam}
=
O\left(\CollCost^{3\tunableparam}\right).
$$
We then have that:
\begin{align*}
&\exp\left(-\frac{m}{8e}\right)
O\left(
n\CollCost^{1/2+2\tunableparam}
\ln\CollCost\,\lg^3 n
\right)\\
&=
O\left(
\exp\left(-\frac{m}{8e}\right)
\exp\left(
\frac{\ln 2}{2}\CollCost^{\tunableparam}
\right)
\CollCost^{1/2+5\tunableparam}
\ln\CollCost
\right)\\
&=O\left(
\exp\left(
\frac{\ln 2}{2}\CollCost^{\tunableparam}
-
\frac{d}{8e}\sqrt{\CollCost}\ln\CollCost
\right)
\CollCost^{1/2+5\tunableparam}
\ln\CollCost
\right).
\end{align*}
where the last line follows by plugging in $m\ge d\sqrt{\CollCost}\ln\CollCost$. Since $0<\tunableparam\leq 1/2$, we have
$\CollCost^\tunableparam\le\sqrt{\CollCost}$ and for
sufficiently large $\CollCost$, we have:
$$
\frac{\ln2}{2}\CollCost^\tunableparam
-
\frac{d}{8e}\sqrt{\CollCost}\ln\CollCost
\le
-\frac{d}{16e}\sqrt{\CollCost}\ln\CollCost.
$$
The corresponding exponentially-small term kills the 
$\CollCost^{1/2+5\tunableparam}\ln\CollCost$ factor, so the expected
contribution after slot $g$ in Equation \ref{eqn:exp-L} is $O(1)$, and so we have shown:
$$
\mathbb{E}[\text{latency}]
\le
\mathbb{E}[L]
\le
g+O(1)
=
O\left(
\CollCost^{1/2+\tunableparam}\ln\CollCost
\right),
$$
as claimed.
\end{proof}

%%%%%%%%%%%%%%%%%%%%%%%%%%%%%%%%%%%%%%%%%%%%%%%%%%%%%
%%%%%%%%%%%%%%%%%%%%%%%%%%%%%%%%%%%%%%%%%%%%%%%%%%%%%

\subsection{Collision Cost Analysis}

We first bound the collision cost contributed by low- and
good-contention slots over the entire execution. For high-contention slots, we split the execution
at the end of the first batch's initial growing phase. Reaching high
contention during that phase is unlikely, and the latency analysis
already bounds both the probability of surviving the phase and the
expected duration of any continued execution. As in the latency analysis, $L$ denotes the slot of the first success
and $g=m(\lfloor\lg w_0\rfloor+1)$ is the full prescribed length of
the first growing phase.  

\begin{lemma}\label{lem:collision-low-good}
The expected collision cost contributed by slots with low or good
contention is $O(\sqrt{\CollCost})$.
\end{lemma}

\begin{proof}
Fix any reachable pre-transmission history $h_t$ under which slot
$t$ has low or good contention. Conditioned on $h_t$, the active
packets and their sending probabilities are fixed, and their
transmission decisions in slot $t$ are independent. Let
$\lambda=\Con(t)$. Then:
$$
0\le\lambda\le\frac{\beta}{\sqrt{\CollCost}}\le\frac12.
$$

Let $P_{\mathrm{col}}(t\mid h_t)$ and
$P_{\mathrm{suc}}(t\mid h_t)$ denote the conditional probabilities
of a collision and a success in slot $t$, respectively.
Lemmas~\ref{lem:collision-probability} and~\ref{lem:bender} give:
$P_{\mathrm{col}}(t\mid h_t)=O(\lambda^2)$
and $ P_{\mathrm{suc}}(t\mid h_t) \ge \lambda e^{-2\lambda}
$. The latter inequality implies:
$$
\lambda^2
\le
\lambda e^{2\lambda}P_{\mathrm{suc}}(t\mid h_t).
$$
The conditional expected collision-cost contribution of slot $t$
therefore satisfies:
$$
\begin{aligned}
\CollCost P_{\mathrm{col}}(t\mid h_t)
&=
O(\CollCost\lambda^2)\\
&=
O(\CollCost\lambda e^{2\lambda})
P_{\mathrm{suc}}(t\mid h_t)\\
&=
O(\sqrt{\CollCost})
P_{\mathrm{suc}}(t\mid h_t),
\end{aligned}
$$
where the last bound uses
$\CollCost\lambda\le\beta\sqrt{\CollCost}$ and
$e^{2\lambda}\le e$.

We now sum these contributions over histories and slots. For each
slot $t$, let $Z_t=\CollCost$ if the execution
reaches slot $t$, the slot has low or good contention, and a collision
occurs; otherwise, let $Z_t=0$. Let $F_t$ be the event that the first
success occurs in slot $t$ and that slot has low or good contention.

Let $\mathcal{H}_t$ be the set of reachable pre-transmission histories
under which slot $t$ has low or good contention. Conditioning on which
history in $\mathcal{H}_t$ occurs gives:
$$
\mathbb{E}[Z_t]
=
\sum_{h_t\in\mathcal{H}_t}
\Pr(h_t)\,
\CollCost P_{\mathrm{col}}(t\mid h_t).
$$
Applying the bound above to each term:
$$
\mathbb{E}[Z_t]
=
O(\sqrt{\CollCost})
\sum_{h_t\in\mathcal{H}_t}
\Pr(h_t)\,
P_{\mathrm{suc}}(t\mid h_t).
$$

Every history in $\mathcal{H}_t$ reaches slot $t$, so no success has
occurred before that slot. Hence, if a success occurs in slot $t$
under one of these histories, it is necessarily the first success.
Therefore:
$$
\sum_{h_t\in\mathcal{H}_t}
\Pr(h_t)\,
P_{\mathrm{suc}}(t\mid h_t)
=
\Pr(F_t).
$$
It follows that $
\mathbb{E}[Z_t]
=
O(\sqrt{\CollCost})\Pr(F_t)$.
Finally, the events $F_t$ are pairwise disjoint, since there can be
only one first successful slot. Thus, $
\sum_{t\ge1}\Pr(F_t)\le1$.
Since the variables $Z_t$ are nonnegative, summing over all slots gives:
$$
\begin{aligned}
\mathbb{E}\left[\sum_{t\ge1}Z_t\right]
&=
\sum_{t\ge1}\mathbb{E}[Z_t] =
O(\sqrt{\CollCost})
\sum_{t\ge1}\Pr(F_t)
=
O(\sqrt{\CollCost}),
\end{aligned}
$$
as claimed.
\end{proof}

%%%%%%%%%%%%%%%%%%%%%%%%%%%%%%%%%%%%%%%%%%%%%
%%%%%%%%%%%%%%%%%%%%%%%%%%%%%%%%%%%%%%%%%%%%%

\begin{lemma}\label{lem:collision-high}
Assume $\CollCost\ge K(\lg^{1/\tunableparam} n)$, where $K>0$
is a sufficiently large constant depending only on $\tunableparam$. For sufficiently large $d$,
the expected collision cost contributed by high-contention slots
is $O(1)$.
\end{lemma}

\ifshowproofs
\begin{proof}
We use the same choice of $K$ as in
Lemma~\ref{lem:first-growing-survival}, so
$K^{\tunableparam}\ge2$ and
$\lg\CollCost\le\CollCost^{\tunableparam}$ whenever $\CollCost\ge K$.
Since $n\ge2$, the condition
$\CollCost\ge K(\lg^{1/\tunableparam} n)$ gives $\CollCost\ge K$
and $\CollCost^{\tunableparam}\ge2\lg n$. Hence:
\begin{align}
\frac{n}{w_0}
=
2^{\lg n-\CollCost^{\tunableparam}}
\le
2^{-\CollCost^{\tunableparam}/2}
\le
2^{-(\lg\CollCost)/2}
=
\frac{1}{\sqrt{\CollCost}}\label{eqn:total-con}.
\end{align}

Let $H_{\le g}$ and $H_{>g}$ be the collision costs contributed by
high-contention slots at or before $g$ and after $g$, respectively.
Each is $\CollCost$ times the number of collisions in its specified slots.

\medskip
\noindent\textbf{Slots at or before $g$.}
Let $E$ be the event that the execution reaches a high-contention
slot at or before $g$. On this event, let $t_1$ be the first such
slot. No reached slot among the first $m$ has high contention:
every packet activated by then is still in its first chunk, so its
contribution is $1/w_0$ and, by Equation \ref{eqn:total-con}, the total is at most
$n/w_0\le1/\sqrt{\CollCost}$. Hence $t_1>m$.

Fix any $t\in[t_1-m,t_1)$. Through slot $g$, every activated packet
is still in its initial growing phase. A packet active at $t$
doubles its probability at most once between $t$ and $t_1$, since
$t_1-t\le m$ and its chunk boundaries are $m$ slots apart. A packet
activated after $t$ is still in its first chunk at $t_1$. Thus:
$$
\frac{\beta}{\sqrt{\CollCost}}
<\Con(t_1)
\le2\Con(t)+\frac{n}{w_0}
\le2\Con(t)+\frac1{\sqrt{\CollCost}}.
$$
Since $t_1$ is the first high-contention slot, this implies:
$$
\frac{\alpha}{\sqrt{\CollCost}}
\le\frac{\beta-1}{2\sqrt{\CollCost}}
<\Con(t)
\le\frac{\beta}{\sqrt{\CollCost}},
$$
where the first inequality uses $\beta\ge8\alpha$ and $\alpha\ge1$.
Therefore, all $m$ slots in $[t_1-m,t_1)$ have good contention, and 
they precede a reached slot, so none contains a success.

Consequently, the event $E$ requires encountering at least $m$
good-contention slots without a success. By
Lemma~\ref{lem:success-good-contention}, we have:
$$
\begin{aligned}
\Pr(E)
&\le\left(1-\frac{\alpha}{e\sqrt{\CollCost}}\right)^m\\
&\le\CollCost^{-\alpha d/e},
\end{aligned}
$$
where the second inequality follows from $1-x\le e^{-x}$ and
$m\ge d\sqrt{\CollCost}\ln\CollCost$. On $E$, at most $g$ slots can contribute to $H_{\le g}$, each contributing
at most $\CollCost$; outside $E$, $H_{\le g}=0$. Choose $d$ sufficiently large to satisfy all earlier requirements
and $\alpha d/e\ge3$. Since $\tunableparam\le1/2$, we have:
$$
\begin{aligned}
\mathbb{E}[H_{\le g}]
&\le\CollCost g\Pr(E)\\
&=O\left(\CollCost^{3/2+\tunableparam-\alpha d/e}
\ln\CollCost\right)\\
&=O\left(\frac{\ln\CollCost}{\CollCost}\right)
=O(1).
\end{aligned}
$$

\medskip
\noindent\textbf{Slots after $g$.}
If $L\le g$, then the first success occurs during the initial
growing phase, so the execution terminates by slot $g$ and $H_{>g}=0$.

It remains to consider the event $L>g$. Conditioned on any
pre-transmission history $h_{g+1}$ reached after surviving the
initial growing phase, we upper-bound $H_{>g}$ by charging the
per-collision cost $\CollCost$ to every remaining slot, whether or not it has
high contention or contains a collision. Lemma~\ref{lem:latency-continuation} gives:
$$
\begin{aligned}
\mathbb{E}[H_{>g}\mid h_{g+1}]
&\le\CollCost\,\mathbb{E}[L-g\mid h_{g+1}]\\
&=O\left(n\CollCost^{3/2+2\tunableparam}
\ln\CollCost\,\lg^3 n\right).
\end{aligned}
$$
Let $\mathcal{H}_{g+1}$ be the set of reachable pre-transmission
histories for slot $g+1$. Since the execution reaches slot $g+1$
exactly when $L>g$:
$$
\sum_{h_{g+1}\in\mathcal{H}_{g+1}}
\hspace{-8pt}\Pr(h_{g+1})
=
\Pr(L>g).
$$
Therefore, conditioning on the history at slot $g+1$:
\begin{align}
\mathbb{E}[H_{>g}]
&=
\hspace{-6pt}\sum_{h_{g+1}\in\mathcal{H}_{g+1}}
\hspace{-9pt}\Pr(h_{g+1})
\,\mathbb{E}[H_{>g}\mid h_{g+1}]\nonumber\\
&\le
\Pr(L>g)\,
O\left(
n\CollCost^{3/2+2\tunableparam}
\ln\CollCost\,\lg^3 n
\right)\nonumber\\
&\le\exp\left(-\frac{m}{8e}\right)
O\left(
n\CollCost^{3/2+2\tunableparam}
\ln\CollCost\,\lg^3 n
\right)\label{eqn:EH},
\end{align}
where the last line follows by Lemma~\ref{lem:first-growing-survival}. 

We next bound the remaining factors in terms of $\CollCost$.
From ${n}/{w_0}\le {1}/{\sqrt{\CollCost}}$ 
and $w_0=2^{\CollCost^{\tunableparam}}$, we have
$n\le {2^{\CollCost^{\tunableparam}}}/{\sqrt{\CollCost}}$, which implies that 
$\lg n
\le \CollCost^{\tunableparam}$, which in turn implies 
$\lg^3 n
\le
\CollCost^{3\tunableparam}$. Substituting these two bounds into Equation \ref{eqn:EH} gives:
\begin{align*}
\mathbb{E}[H_{>g}]
&\le
\exp\left(-\frac{m}{8e}\right)
O\left(
\frac{2^{\CollCost^{\tunableparam}}}{\sqrt{\CollCost}}
\CollCost^{3/2+2\tunableparam}
\ln\CollCost\,
\CollCost^{3\tunableparam}
\right)\\
&=
\exp\left(-\frac{m}{8e}\right)
O\left(
2^{\CollCost^{\tunableparam}}
\CollCost^{1+5\tunableparam}
\ln\CollCost
\right)\\
&=O(1),
\end{align*}
where the last line follows from $m\ge d\sqrt{\CollCost}\ln\CollCost$ and the fact that the exponential decay dominates the remaining polynomial and
logarithmic factors. Since $\mathbb{E}[H_{\le g}]=O(1)$ and
$\mathbb{E}[H_{>g}]=O(1)$, the expected collision cost contributed
by high-contention slots is $O(1)$.
\end{proof}
\fi
%%%%%%%%%%%%%%%%%%%%%%%%%%%%%%%%%%%%%%%%%%%%%
%%%%%%%%%%%%%%%%%%%%%%%%%%%%%%%%%%%%%%%%%%%%%

\begin{lemma}\label{lem:collision-dynamic}
There is a constant $K>0$, depending only on $\tunableparam$, such
that the expected collision cost of \mainAlgorithmAcrynom in the
dynamic setting satisfies the following bounds.
\begin{enumerate}[leftmargin=14pt]
    \item If $\CollCost\ge K(\lg^{1/\tunableparam} n)$, then the
    expected collision cost is $O(\sqrt{\CollCost})$.

    \item If $\CollCost<K(\lg^{1/\tunableparam} n)$, then the
    expected collision cost is
    $O\hspace{-2pt}\left(n\log^{\Theta(1/\tunableparam)}n\right)$.
\end{enumerate}
\end{lemma}

\ifshowproofs
\begin{proof}
Choose $K$ sufficiently large to satisfy the requirements of
Lemma~\ref{lem:collision-high}. Our analysis treats two cases.\medskip

\noindent\textbf{Case 1:
$\CollCost\ge K(\lg^{1/\tunableparam} n)$.}
Every slot has low, good, or high contention. Adding the bounds
from Lemmas~\ref{lem:collision-low-good} and~\ref{lem:collision-high}
shows that the expected collision cost over the execution is:
$$
O(\sqrt{\CollCost})+O(1)
=
O(\sqrt{\CollCost}).
$$
\noindent\textbf{Case 2:
$\CollCost<K(\lg^{1/\tunableparam} n)$.}
Since $K$ and $\tunableparam$ are fixed constants, we have:
$\CollCost=O(\lg^{1/\tunableparam} n)$. Each slot contributes at most
$\CollCost$, so the execution's collision cost is at most $\CollCost L$. Applying
Lemma~\ref{lem:latency-continuation} with $s=1$ and averaging over
the possible initial histories gives that the expected collision cost over the execution is upper-bounded by:
$$\CollCost\,\mathbb{E}[L]\\
=
O\left(
n\CollCost^{3/2+2\tunableparam}
\ln\CollCost\,\lg^3 n
\right).$$
The bound on $\CollCost$ gives:
$$
\CollCost^{3/2+2\tunableparam}
=
O\left(
\lg^{3/(2\tunableparam)+2}n
\right).
$$
Thus, the expected collision cost over the execution is bounded by:
$$
\begin{aligned}
&O\left(
n\lg^{3/(2\tunableparam)+2}n
\cdot\lg n\cdot\lg^3 n
\right)\\
&=
O\left(
n\lg^{3/(2\tunableparam)+6}n
\right)\\
&=
O\left(
n\log^{\Theta(1/\tunableparam)}n
\right),
\end{aligned}
$$
since $3/(2\tunableparam)+6=\Theta(1/\tunableparam)$ for
$0<\tunableparam\le1/2$.
\end{proof}
\fi

%%%%%%%%%%%%%%%%%%%%%%%%%%%%%%%%%%%%%%%%%%%%%%%%%%%%%%
%%%%%%%%%%%%%%%%%%%%%%%%%%%%%%%%%%%%%%%%%%%%%%%%%%%%%%

{\textsc{Theorem}~\ref{thm:dynamic}.} {\it Fix a constant $0<\tunableparam\le1/2$ and choose $d$ sufficiently
large. There is a constant $K>0$, depending only on
$\tunableparam$, such that against any adaptive adversary satisfying
our model, \mainAlgorithmAcrynom solves the dynamic wakeup problem with the following
bounds.
\begin{enumerate}[leftmargin=14pt]
    \item If $\CollCost\ge K(\lg^{1/\tunableparam} n)$, then the
    expected latency is
    $O(\CollCost^{1/2+\tunableparam}\ln\CollCost)$
    and the expected collision cost is $O(\sqrt{\CollCost})$.\smallskip

    \item If $\CollCost<K(\lg^{1/\tunableparam} n)$, then the
    expected latency and expected collision cost are both
    $O(n\log^{\Theta(1/\tunableparam)}n)$.
\end{enumerate}
}

\begin{proof}
Choose $K$ sufficiently large to satisfy the requirements of
Lemmas~\ref{lem:latency-large-cost} and~\ref{lem:collision-dynamic}.

\medskip
\noindent\textbf{Case 1:
$\CollCost\ge K(\lg^{1/\tunableparam} n)$.} By
Lemma~\ref{lem:latency-large-cost}, the expected latency is 
$O(\CollCost^{1/2+\tunableparam}\ln\CollCost)$. By Lemma~\ref{lem:collision-dynamic}, the expected collision cost is $O(\sqrt{\CollCost})$.

\medskip
\noindent\textbf{Case 2:
$\CollCost<K(\lg^{1/\tunableparam} n)$.}
Again, since $K$ and $\tunableparam$ are fixed constants, we have $\CollCost=O(\lg^{1/\tunableparam} n)$. Applying Lemma~\ref{lem:latency-continuation} with $s=1$ and
averaging over the possible initial histories gives:
$$
\mathbb{E}[L]
=
O\left(
n\CollCost^{1/2+2\tunableparam}
\ln\CollCost\,\lg^3 n
\right).
$$
The bound on $\CollCost$ implies:
$$
\CollCost^{1/2+2\tunableparam}
=
O\left(
\lg^{1/(2\tunableparam)+2}n
\right).
$$
Therefore, we have:
$$
\begin{aligned}
\mathbb{E}[L]
&=
O\left(
n\lg^{1/(2\tunableparam)+2}n
\cdot\lg n\cdot\lg^3 n
\right)\\
&=
O\left(
n\lg^{1/(2\tunableparam)+6}n
\right)\\
&=
O\left(
n\log^{\Theta(1/\tunableparam)}n
\right),
\end{aligned}
$$
since $1/(2\tunableparam)+6=\Theta(1/\tunableparam)$ for
$0<\tunableparam\le1/2$. By Lemma~\ref{lem:collision-dynamic}, the expected collision cost is:
$O(n\log^{\Theta(1/\tunableparam)}n)$, which completes the argument.
\end{proof}

%%%%%%%%%%%%%%%%%%%%%%%%%%%%%%%%%%%%%%%%%%%%%%%%%%%%%% 
%%%%%%%%%%%%%%%%%%%%%%%%%%%%%%%%%%%%%%%%%%%%%%%%%%%%%% 

\section{Lower Bound}\label{sec:lower-bound}

We first prove a tradeoff between latency and collision cost, then explain the classical dynamic-wakeup lower bound when the per-collision cost $\CollCost$ is fixed and $n$ grows.\medskip

\noindent{\bf Algorithm class and metrics.}
We consider \defn{age-based batch-fair} algorithms. For each value of
$\CollCost$, such an algorithm $A$ specifies probabilities:
$$
p_1(\CollCost),p_2(\CollCost),p_3(\CollCost),\ldots.$$
In its $j$-th slot after activation, a packet transmits with probability
$p_j(\CollCost)$, using fresh randomness independent of the other packets.
The probabilities depend only on $\CollCost$ and local age, not on $n$,
global time, or previous transmission outcomes. Thus, packets activated
together use the same probability in each slot, which is the meaning of
batch fairness here. \mainAlgorithmAcrynom belongs to this class for each fixed choice of
its algorithm parameters.

Let $\mathcal{L}_A(n,\CollCost)$ denote the worst-case expected latency
of $A$, and let $\mathcal{K}_A(n,\CollCost)$ denote its worst-case
expected collision cost, where the worst case is over adversaries
satisfying our model and activating at most $n$ packets. Packets need
not all be activated before the first success. Define:
$$
\mathcal{M}_A(n,\CollCost)
=
\max\left\{
\mathcal{L}_A(n,\CollCost),
\mathcal{K}_A(n,\CollCost)
\right\}.
$$
An execution with no success has infinite latency. The lower bounds
below use activation schedules fixed in advance, so they also apply
against the adaptive adversaries allowed by our model.

\subsection{Tradeoff: Latency \& Expected Collision Cost}

\begin{lemma}\label{lem:lb-collision}
For every age-based batch-fair algorithm $A$, every $n\ge2$, and every
$\CollCost$, if
$\mathcal{L}_A(n,\CollCost)<\infty$, then:
$$
\mathcal{L}_A(n,\CollCost)
\mathcal{K}_A(n,\CollCost)
\ge\frac{\CollCost}{4}.
$$
\end{lemma}
\begin{proof}
Consider this activation schedule: we activate two of the $n$ packets in slot $1$ and no others.  Recall that $L$ denotes the latency,
with $L=\infty$ if no success occurs, and let $X$ be the collision
cost of the execution.  

For this activation schedule, we have $\mathbb{E}[L]\le\mathcal{L}_A(n,\CollCost)<\infty$, where the finiteness of$ \mathcal{L}_A(n,\CollCost)$ is an assumption in the lemma. Consequently, a success occurs with probability $1$: otherwise,
$L=\infty$ with positive probability, which would imply
$\mathbb{E}[L]=\infty$, contradicting the bound above.

%%%%%%

Let $a_t=\Pr(L\ge t)$, which is the probability that the execution
reaches slot $t$ without a success in any earlier slot. Since both
packets were activated in slot $1$, conditioned on reaching slot $t$,
each transmits independently with probability $p_t(\CollCost)$.
Therefore:
$$
\Pr(L=t)
=
a_t\,2p_t(\CollCost)(1-p_t(\CollCost)).
$$
Since a success occurs with probability $1$, the events
$\{L=t\}$, for $t\ge1$, account for all executions. Hence:
\begin{equation}
1
=
\sum_{t\ge1}\Pr(L=t)
=
\sum_{t\ge1}a_t\,2p_t(\CollCost)(1-p_t(\CollCost))
\le
2\sum_{t\ge1}a_t p_t(\CollCost).
\label{eqn:fraction}
\end{equation}

%%%%%%

We next express the expected latency and expected collision cost in terms of the probabilities $a_t$. First, since $L$ is a positive integer-valued random variable,
the tail-sum formula gives:
$$
\mathbb{E}[L]
=
\sum_{t\ge1}\Pr(L\ge t)
=
\sum_{t\ge1}a_t.
$$
Second, let $C_t$ be the indicator that slot $t$ is reached and
contains a collision. Then $X=\CollCost\sum_{t\ge1}C_t$.

Conditioned on reaching slot $t$, both packets transmit independently
with probability $p_t(\CollCost)$, so the conditional probability of
a collision is $p_t(\CollCost)^2$. Therefore:
$$
\Pr(C_t=1)
=
\Pr(L\ge t)\,p_t(\CollCost)^2
=
a_t p_t(\CollCost)^2.
$$
Taking expectations and summing over the slots gives:
$$
\begin{aligned}
\mathbb{E}[X]
&=
\CollCost\sum_{t\ge1}\mathbb{E}[C_t]
= \CollCost\sum_{t\ge1}a_t p_t(\CollCost)^2.
\end{aligned}
$$
Equivalently, we have
${\mathbb{E}[X]}/{\CollCost}
=
\sum_{t\ge1}a_t p_t(\CollCost)^2$.
%%%%%%%%%%%
If $\mathbb{E}[X]=\infty$, then
$\mathcal{K}_A(n,\CollCost)=\infty$ and the claim is immediate. Otherwise, Equation \ref{eqn:fraction} gives $
{1}/{2}
\le
\sum_{t\ge1}a_t p_t(\CollCost)$.
Squaring both sides gives:
$$
\frac14
\le
\left(\sum_{t\ge1}a_t p_t(\CollCost)\right)^2.
$$
Applying Cauchy--Schwarz \cite{steele2004cauchy} to the two sequences
$\left(\sqrt{a_t}\right)_{t\ge1}$ and
$\left(\sqrt{a_t}\,p_t(\CollCost)\right)_{t\ge1}$ gives: 
$$
\begin{aligned}
\frac14
&\le
\left(\sum_{t\ge1}a_t p_t(\CollCost)\right)^2\\
&=
\left(
\sum_{t\ge1}
\sqrt{a_t}
\left(\sqrt{a_t}\,p_t(\CollCost)\right)
\right)^2\\
&\le
\left(\sum_{t\ge1}a_t\right)
\left(\sum_{t\ge1}a_t p_t(\CollCost)^2\right)\\
&=
\frac{\mathbb{E}[L]\mathbb{E}[X]}{\CollCost}.
\end{aligned}
$$
Since $\mathcal{L}_A(n,\CollCost)$ and
$\mathcal{K}_A(n,\CollCost)$ are the worst-case expected latency
and expected collision cost, respectively, over all valid activation
schedules, 
$\mathbb{E}[L]\le\mathcal{L}_A(n,\CollCost)$ and 
$\mathbb{E}[X]\le\mathcal{K}_A(n,\CollCost)$. Therefore:
$\mathcal{L}_A(n,\CollCost)
\mathcal{K}_A(n,\CollCost)
\ge
{\CollCost}/{4}$, as claimed.
\end{proof}

{\textsc{Theorem}~\ref{thm:dynamic-lower-bound}.}  {\it 
For every age-based batch-fair algorithm $A$, every $n\ge2$, and every
$\CollCost$,  $\mathcal{M}_A(n,\CollCost)
\ge
{\sqrt{\CollCost}}/{2}$.
}

\begin{proof}
If $\mathcal{M}_A(n,\CollCost)$ is not finite, the claim is trivially satisfied. Otherwise, the expected latency is finite, so
Lemma~\ref{lem:lb-collision} and the definition of
$\mathcal{M}_A$ give:
$$
{\CollCost}/{4}
\le
\mathcal{L}_A(n,\CollCost)
\mathcal{K}_A(n,\CollCost)
\le
\mathcal{M}_A(n,\CollCost)^2,
$$
and taking the square root on both sides yields the result.
\end{proof}

\medskip
\noindent{\bf Comparison with \mainAlgorithm.} An $O(\sqrt{\CollCost})$ worst-case expected
collision-cost guarantee requires $\Omega(\sqrt{\CollCost})$
worst-case expected latency. When
$\CollCost\ge K(\lg^{1/\tunableparam} n)$,
Lemmas~\ref{lem:latency-large-cost}
and~\ref{lem:collision-dynamic} give:
$$
\begin{aligned}
\mathcal{K}_{\mathrm{\mainAlgorithmAcrynom}}(n,\CollCost)
&=O\left(\sqrt{\CollCost}\right),\\
\mathcal{L}_{\mathrm{\mainAlgorithmAcrynom}}(n,\CollCost)
&=O\left(\CollCost^{1/2+\tunableparam}\ln\CollCost\right).
\end{aligned}
$$
Thus, in this regime, the maximum of \mainAlgorithmAcrynom's two expectations is
within a factor $O(\CollCost^{\tunableparam}\ln\CollCost)$ of the
best possible for this algorithm class.

%%%%%%%%%%%%%%%%%%%%%%%%%%%%%%%%%%%%%%%%%%%%%%%%%%%%%%%%%%%%%

% we're compressing this section for the conference submission, but I'd like to keep the original.
\ifshowproofs
\subsection{The Classical Dynamic Lower Bound for Fixed Per-Collision Cost}

The classical dynamic-wakeup lower bound requires some care in our
setting. Theorem~6 of Jurdzi\'nski and Stachowiak
\cite{jurdzinski2005probabilistic} proves an
$\Omega(n/\log n)$ bounded-error latency lower bound for dynamic
wakeup when packets have no identifiers, have only local clocks,
and do not know $n$. At first glance, this bound may appear stronger
than our $\Omega(\sqrt{\CollCost})$ lower bound and, in some parameter
regimes, inconsistent with our upper bounds. For completeness, we
record a direct adaptation of their construction to our age-based
batch-fair algorithm class and make explicit how the resulting bound
depends on the per-collision cost $\CollCost$.

When $\CollCost$ is fixed independently of $n$, this adaptation gives
an $\Omega(n/\log n)$ expected-latency lower bound. However, the hidden
constant and the required threshold on $n$ may depend on the chosen
value of $\CollCost$. This dependence matters because packets know
$\CollCost$, so their transmission-probability schedules may change
with it. Thus, when $\CollCost$ is allowed to vary with $n$, the
argument does not give an $\Omega(n/\log n)$ lower bound with a hidden
constant independent of $\CollCost$. Consequently, it neither
contradicts our large-$\CollCost$ upper bounds nor yields an additional
$\Omega(n/\log n)$ term that can be combined with our
$\Omega(\sqrt{\CollCost})$ lower bound.

%%%%%%%%%%%%%%%%%%%%%%%%%%%%%%%%%%%%%%%
%%%%%%%%%%%%%%%%%%%%%%%%%%%%%%%%%%%%%%%

\begin{lemma}\label{lem:lb-dynamic-latency}
Fix an age-based batch-fair algorithm $A$ and a per-collision cost
$\CollCost$. There exist constants $c_0>0$ and $n_0$, which may depend
on $A$ and $\CollCost$, such that for every $n\ge n_0$:
$$
\mathcal{L}_A(n,\CollCost)
\ge c_0\frac{n}{\ln n}.
$$
\end{lemma}
\begin{proof}
We adapt the batch-injection construction in the proof of
Theorem~6 of Jurdzi\'nski and Stachowiak
\cite{jurdzinski2005probabilistic}.
If $p_j(\CollCost)=0$ for every $j$, no packet ever transmits and
latency is infinite. Otherwise, let $j$ be the smallest index with
$p_j(\CollCost)>0$, and write $p=p_j(\CollCost)$. Since the algorithm
and $\CollCost$ are fixed, $j$ and $p$ do not depend on $n$.

For sufficiently large $n$, set:
$$
r=\left\lceil\frac{3\ln n}{p}\right\rceil,
\qquad
T=\left\lfloor\frac{n}{2r}\right\rfloor.
$$
Activate a batch of $2r$ packets in each of slots $1,\ldots,T$, with
no further activations. This uses at most $n$ packets. Partition each
batch into two groups of $r$ packets for the analysis.

No success is possible before slot $j$, since no packet yet has
positive sending probability. For each $j\le t\le T$, the batch
activated in slot $t-j+1$ has local age $j$. Conditioned on reaching
slot $t$, its packets independently transmit with probability $p$.
If each of its two groups contains a transmitter, then slot $t$
has a collision, regardless of the other batches' actions. Thus,
for every history reaching slot $t$, its success probability is
at most:
$$
2(1-p)^r\le2e^{-pr}\le\frac{2}{n^3}.
$$
Summing the probabilities of a first success in slots $1,\ldots,T$
gives:
$$
\Pr(L\le T)\le\frac{2T}{n^3}\le\frac{2}{n^2}.
$$
On $L>T$, latency is at least $T$. Therefore:
$$
\mathcal{L}_A(n,\CollCost)
\ge T\left(1-\frac{2}{n^2}\right).
$$
For sufficiently large $n$, we have $r\le4\ln n/p$ and
$n/(2r)\ge2$, which imply $T\ge pn/(16\ln n)$. Hence:
$$
\mathcal{L}_A(n,\CollCost)
\ge\frac{p n}{32\ln n}.
$$
This proves the lemma. In particular, the coefficient obtained
by this construction depends on $p_j(\CollCost)$.
\end{proof}

%%%%%%%%%%%%%%%%%%%%%%%%%%%%%%%%%%%%%%%
%%%%%%%%%%%%%%%%%%%%%%%%%%%%%%%%%%%%%%%

\noindent{\bf Why holding $\CollCost$ constant matters.}
The lemma first chooses a particular value of $\CollCost$ and then
takes $n$ sufficiently large while keeping that value unchanged.
It does not give a constant or a threshold on $n$ that is independent
of $\CollCost$. Our model supplies $\CollCost$ to the packets, so the
probability sequence can change when $\CollCost$ changes, even though
the packets do not know $n$.

For \mainAlgorithmAcrynom, the first positive sending probability is
$p_1(\CollCost)=1/w_0$. The construction above therefore requires
batches of $\Theta(w_0\ln n)$ packets. In the large-$\CollCost$ 
regime of Theorem \ref{thm:dynamic} where, $\CollCost\ge K(\lg^{1/\tunableparam}n)$, our choice of
$K$ means $\CollCost^\tunableparam\ge2\lg n$, and 
so $w_0=2^{\CollCost^\tunableparam}\ge n^2$. The construction above requires batches of
$\Theta(w_0\ln n)$ packets, which is more than the total population of $n$ packets. Thus, the construction cannot be carried out in this regime.

This does not contradict the lower bound established in
Lemma~\ref{lem:lb-dynamic-latency}. That lemma first chooses a value
of $\CollCost$ and then holds it constant as $n$ grows. In that case,
$w_0$ is also constant, so for sufficiently large $n$ there are enough
packets to form batches of size $\Theta(w_0\ln n)$, as required by the
construction. However, as $n$ grows while $\CollCost$ remains constant,
the condition
$\CollCost\ge K(\lg^{1/\tunableparam}n)$ eventually fails. Thus, the
fixed-$\CollCost$ lower bound applies asymptotically outside the
large-$\CollCost$ regime in which we obtain our sharper upper bound.

%%%%%%%%%%%%%%%

For a concrete example, take $\tunableparam=1/2$ and
$\CollCost=K\lg^2 n$, where $K$ is the constant from our
large-$\CollCost$ case. Since $1/\tunableparam=2$, this places us
exactly at the threshold in Theorem \ref{thm:dynamic}:
$$
\CollCost
=
K(\lg^{1/\tunableparam} n)
=
K\lg^2 n.
$$
The latency bound for this case is $
O(\CollCost^{1/2+\tunableparam}\ln\CollCost)=
O(
\CollCost\ln\CollCost
)$.
Substituting $\CollCost=K\lg^2 n$ gives:
$$
O\left(
K (\lg^2 n)
\ln\left(K\lg^2 n\right)
\right)
=
O\left(
\lg^2 n(\ln\ln n)
\right).
$$
The expected collision cost is even smaller:
$$
O\left(\sqrt{\CollCost}\right)
=
O(\lg n).
$$
Therefore:
$$
\mathcal{M}_{\mathrm{\mainAlgorithmAcrynom}}(n,\CollCost)
=
O\left(
(\lg^2 n)\ln\ln n
\right),
$$
and we note that:
$$
(\lg n)^2\ln\ln n
=
o\left(\frac{n}{\log n}\right).
$$
Thus, if the classical $\Omega(n/\log n)$ lower bound held with a
constant independent of $\CollCost$ even when $\CollCost$ were allowed
to grow with $n$, it would contradict this upper bound. The example
therefore illustrates why that lower bound must be interpreted with
$\CollCost$ held constant while $n$ grows. Thus, the $\Omega(n/\log n)$ latency lower bound cannot be combined
with our $\Omega(\sqrt{\CollCost})$ tradeoff to obtain a single
$\Omega(\sqrt{\CollCost}+n/\log n)$ lower bound with a hidden constant independent of $\CollCost$.

\fi
%%%%%%%%%%%%%%%%%%%%%%%%%%%%%%%%%%%%%%%%%%%%%%%%%%%%%%%%%%%%%%%%%%%%%

\section{Conclusion and Future Work}

We have addressed dynamic wakeup with a per-collision cost and no collision detection. Our results show that contention can be controlled even under arbitrary packet activations while keeping both latency and collision cost bounded. 

In the large-$\CollCost$ regime, the $O(\sqrt{\CollCost})$ expected
collision-cost bound matches the scale of our
$\Omega(\sqrt{\CollCost})$ tradeoff. For the maximum of expected
latency and expected collision cost, our upper bound is within a factor
$O(\CollCost^\tunableparam\ln\CollCost)$ of this lower bound for the
algorithm class we consider. It remains open whether this gap can be
reduced or eliminated. It would also be interesting to extend the lower
bound beyond age-based batch-fair algorithms. % and to study whether similarguarantees can be achieved in the presence of adversarial jamming or other forms of channel disruption.
\medskip

\noindent{\bf Acknowledgements/Generative AI Disclosure.} The research problem, model, algorithmic design, and overall proof
strategy presented here belong to the authors. That said, AI was very 
helpful in simplifying and tightening the analysis. We also used AI to
polish the paper, although the presentation and high-level exposition
are our own. We take full responsibility for any errors.

\bibliographystyle{splncs04}
\bibliography{CR}

\end{document}